\documentclass[a4paper,onecolumn,11pt,accepted=2022-05-20]{quantumarticle}
\pdfoutput=1
\usepackage[utf8]{inputenc}
\usepackage[english]{babel}
\usepackage[T1]{fontenc}

\usepackage{amsmath, amsthm, amsfonts, amssymb}
\usepackage[mathscr]{euscript}
\usepackage{calc}
\usepackage{complexity}
\usepackage{array}
\usepackage{mathdots}
\usepackage{tikz}
\usepackage{tikz-3dplot}
\usepackage{rotating}
\usepackage{pdflscape}
\usepackage[numbers, sort&compress]{natbib}

\usepackage[hidelinks]{hyperref}
\hypersetup{
    colorlinks=true, % false: boxed links; true: colored links
}
\usepackage{fullpage}
\usetikzlibrary{calc,positioning,fit,backgrounds,matrix,decorations.markings,decorations.pathreplacing}
\input{Qcircuit}

\theoremstyle{plain}
\newtheorem{theorem}{Theorem}
\newtheorem{proposition}[theorem]{Proposition}
\newtheorem{lemma}[theorem]{Lemma}
\newtheorem{corollary}[theorem]{Corollary}
\newtheorem{conjecture}[theorem]{Conjecture}

\newcommand{\horizontally}{\hspace{5pt}}
\newcommand{\abs}[1]{\left\lvert#1\right\rvert}
\let\P\relax

\let\R\relax
\let\class\relax

\newcommand{\CC}{\mathcal{C}}
\newcommand{\P}{\mathcal{P}}
\newcommand{\clifford}{\mathcal{C}}
\newcommand{\NOT}{\operatorname{NOT}}
\newcommand{\CNOT}{\operatorname{CNOT}}
\newcommand{\CSIGN}{\operatorname{CZ}}
\newcommand{\SWAP}{\operatorname{SWAP}}
\newcommand{\R}{\operatorname{R}}
\newcommand{\tgate}{\operatorname{T}}
\newcommand{\tfour}{\tgate_4}

\newcommand{\CNOTNOT}{\operatorname{CNOTNOT}}
\newcommand{\Fredkin}{\operatorname{Fredkin}}

\newcommand{\class}[1]{ \left< #1 \right>}
\newcommand{\ALL}{\mathsf{ALL}}
\newcommand{\ftwo}{\mathbb{F}_2}
\newcommand{\subring}{\mathsf{S}}
\newcommand{\ring}{\mathsf{R}}
\newcommand{\Tstar}{\mathcal{M}}
\let\G\relax
\newcommand{\G}{\mathcal{G}}

\title{The Classification of Clifford Gates over Qubits}
\author{Daniel Grier}
\affiliation{University of Waterloo, Cheriton School of Computer Science}
\email{daniel.grier@uwaterloo.ca}
\author{Luke Schaeffer}
\affiliation{University of Waterloo, Department of Combinatorics and Optimization}
\email{lrschaeffer@gmail.com}
\begin{document}
% \maketitle

%!TEX root = ../full_paper.tex

\begin{abstract}

We examine the following problem:  given a collection of Clifford gates, describe the set of unitaries generated by circuits composed of those gates. Specifically, we allow the standard circuit operations of composition and tensor product, as well as ancillary workspace qubits as long as they start and end in states uncorrelated with the input, which rule out common ``magic state injection'' techniques that make Clifford circuits universal.  We show that there are exactly 57 classes of Clifford unitaries and present a full classification characterizing the gate sets which generate them.  This is the first attempt at a quantum extension of the classification of reversible classical gates introduced by Aaronson et al., another part of an ambitious program to classify all quantum gate sets.

The classification uses, at its center, a reinterpretation of the tableau representation of Clifford gates to give circuit decompositions, from which elementary generators can easily be extracted.  The 57 different classes are generated in this way, 30 of which arise from the single-qubit subgroups of the Clifford group.  At a high level, the remaining classes are arranged according to the bases they preserve.  For instance, the CNOT gate preserves the $X$ and $Z$ bases because it maps $X$-basis elements to $X$-basis elements and $Z$-basis elements to $Z$-basis elements.  The remaining classes are characterized by more subtle tableau invariants; for instance, the $\tfour$ and phase gate generate a proper subclass of $Z$-preserving gates.
\end{abstract}

%!TEX root = ../full_paper.tex
\section{Introduction}
A common thread throughout quantum computing is the manner in which a few elementary gates often suffice for universal computation.  This ``pervasiveness of universality'' is explored in recent work of Aaronson, Grier, and Schaeffer \cite{ags:2015}.  There, the authors give a complete classification of \emph{classical} reversible gates in terms of the functions over bits they generate and find that a rich structure emerges.  

Of course, the ultimate goal would be a complete classification of \emph{quantum} gate sets based on the functions over qubits they generate.  Unfortunately, not even a full classification of the subgroups of a three-qubit system is known.\footnote{The difficulty in classifying the subgroups of $\mathrm{SU}(N)$ arises not from the infinite classes but from the finite ones.  In fact, even the finite subgroups of $\mathrm{SU}(5)$  remain unclassified. This motivates our focus on finite, discrete classes such as the Clifford group.}  Since each class of gates is a subgroup, this suggests that a complete classification remains out of reach. This might be surprising given how well we understand random gate sets, and even those that contain particular gates such as CNOT \cite{barenco, shi:gate}.  However, a full classification begets a complete understanding of \emph{all} possible behaviors, despite their strangeness or rarity (see, for example, the sporadic gate sets in the lattice of Aaronson et al.\ \cite{ags:2015}).  Nevertheless, there has been some encouraging progress on classification problems:  for Hamiltonians, Bouland, Man\u{c}inska, and Zhang \cite{bmz:2015} classified all 2-qubit commuting Hamiltonians, while Childs et al.\ \cite{childs:2010} characterized all 2-qubit Hamiltonians when restricted to circuits over two qubits; for linear optics, Aaronson and Bouland \cite{aarbouland} completed a classification for linear optics of $2$-mode beamsplitters; and for Clifford+$T$ circuits, Amy, Glaudell, and Ross \cite{amyglaudellross:2020} give a type of classification based on the elements appearing in their representations as unitary matrices.
% which relied heavily on the characterization \cite{grimus:2014} of the finite subgroups of $\mathrm{SU}(3)$, underscoring the difficulty of quantum gate classification.
%  Nevertheless, there has been some encouraging progress on classification problems, in particular, on classifying Hamiltonians (which can be applied for any period of time), rather than discrete gate sets.   For instance, Bouland, Man\u{c}inska, and Zhang \cite{bmz:2015} recently classified all 2-qubit commuting Hamiltonians, while Childs et al.\ \cite{childs:2010} characterized all 2-qubit Hamiltonians when restricted to circuits over two qubits.  Additionally, Aaronson and Bouland \cite{aarbouland} completed a classification for linear optics of $2$-mode beamsplitters, which relied heavily on the characterization \cite{grimus:2014} of the finite subgroups of $\mathrm{SU}(3)$, underscoring the difficulty of quantum gate classification.

This paper contributes a new classification of quantum gate sets by giving a complete classification of the Clifford gates, the set of unitaries normalizing the Pauli group.  To provide some context, the Clifford gates are generated by the CNOT gate, the Hadamard gate, and the $\frac{\pi}{4}$-phase gate.  It is not hard to see that the Clifford operations on $n$ qubits are a discrete, finite set, so it has always been widely assumed that they do not suffice for universal quantum computation.  Indeed, Gottesman and Knill \cite{gottesman:1998} showed that they could be efficiently simulated with a classical computer, using a binary matrix representation of states called a \emph{stabilizer tableau}.  Using a slightly larger version of these tableaux, Aaronson and Gottesman \cite{ag} were able to further improved the efficiency of measurements in this algorithm.\footnote{In terms of complexity classes, Aaronson and Gottesman \cite{ag} show that Clifford circuits can be simulated in the class $\oplus \mathsf{L}$ (pronounced ``parity ell''). This class is most easily understood as capturing those problems reducible to solving linear equations mod 2.  To be clear, we have that $\oplus \mathsf{L} \subseteq \mathsf{P}$.}  In fact, a reinterpretation of their tableau representation is integral to our classification.

% Aaronson and Gottesman \cite{ag} improved this result to show that Clifford circuits can be simulated in $\oplus \mathsf{L}$, for which a complete problem is the solution of a linear system over $\ftwo$.\footnote{In that work, Aaronson and Gottesman proposed the \emph{tableau representation} of a Clifford gate, which we reinterpret to serve as one of the principal components of our classification theorem.  See Section~\ref{sec:tableaux} for those details.}

Clifford circuits are somewhat remarkable in that they may in fact be integral to our eventual development of a general-purpose quantum computer.  For instance, the stabilizer formalism, which tracks state evolution through conjugated Pauli elements, underlies many of the important quantum error correcting codes \cite{gottesman:thesis}.  In fact, the Clifford operations are exactly those operations which can be easily computed transversally in many fault-tolerant schemes of quantum computing (e.g., the Shor code \cite{shor:1996} or the [[7,1,3]] Steane code \cite{steane}).

Our model is motivated in part by the use of Clifford circuits as subroutines of a general quantum computation, much like the transversal gates in a fault-tolerant scheme.  We regard the creation of complicated ancilla states as an inherently difficult task, and therefore require that all ancillary qubits used during the computation be returned to their initial state at the end of the computation.  This restriction eliminates schemes in which much of the difficulty of the computation is offloaded to the creation of ``magic states'' which are subsequently consumed by the computation \cite{bravyi:2005, rbb, shor:1996}.  Unlike these schemes in which the ancillas boost weak gate sets to computational universality, Clifford operations cannot be boosted in our model to generate anything outside the Clifford group.\footnote{This was first noticed by Anderson \cite{anderson:2012}.}

We also regard the classification of Clifford gates as an important step towards a full classification of quantum gate classes.  Although the complete inclusion lattice for general quantum gates will be significantly more complicated than the one we present for the Clifford gates, the classes described here provide a testbed for the techniques used for general quantum gates.  This is due to the fact that our lattice for Clifford gates must appear as a sublattice in the complete quantum gate classification.  This situation contrasts with the reversible gate classification of Aaronson et al., in which much of the complexity of the lattice is due to the fact that only $\ket{0}$ and $\ket{1}$ ancillas were allowed.  Indeed, we show in Appendix~\ref{app:classical_quantum} that the classical reversible classification collapses significantly under quantum ancillas.  Because we allow for arbitrary ancillas in our model, our classification does not suffer from the same issue. 

\subsection{Results}
We wish to determine the set of Clifford operations that can be realized as circuits consisting of gates from a given gate set.  Let us briefly explain the circuit building operations we allow (full details and justification in Section~\ref{sec:classes}).  First, we can combine gates in series or parallel, i.e., their composition or tensor product.  We also assume for simplicity that swapping qubits is allowed at any point in the circuit. Each circuit also has access to arbitrary quantum ancillas provided that they are returned to their initial states by the end of the computation. Finally, we adopt the standard practice of ignoring global phase in circuits. Under this model, our main result is the classification of Clifford gates below:
\begin{theorem}
Any set of Clifford gates generate one of 57 distinct classes of Clifford operations.  There are 30 classes (depicted in Figure~\ref{fig:degeneratelattice}) generated by single-qubit gates.  The remaining 27 non-degenerate classes are shown in Figure~\ref{fig:nondegeneratelattice}.  Notation for the generators of the classes depicted in those diagrams is given in Section~\ref{sec:gates}. 
\end{theorem}

We list some consequences and highlights of the classification below:
{\renewcommand\labelenumi{(\theenumi)}
\begin{enumerate}
\item {\bf Invariants.} Every class can be defined by a collection of invariants, i.e., properties of the Clifford gates which are preserved under our circuit building operations. Formally, we define each invariant based on the tableau representation of the Clifford gate (see sections \ref{sec:tableaux} and \ref{sec:invariants}). We now describe the broad themes behind the main invariants of the classification. First, there is a three-fold symmetry in the classification, corresponding to the symmetry of the $X$, $Y$, and $Z$ elements of the Pauli group. For example, the CNOT gate behaves classically in the $X$-basis and $Z$-basis (i.e., it permutes the four $2$-qubit $X$-basis vectors, and likewise in the $Z$-basis), but not the $Y$-basis. By symmetry, there are gates like CNOT which are classical in each pair of bases. There is even a nontrivial class of gates which act classical in all three bases (up to sign). %permutes the $X$-basis, and permutes the $Z$-basis, but not the $Y$-basis. Naturally, there are gates corresponding to CNOT for other choices of bases. In fact, there is a nontrivial class of gates that act as permutations in all three bases (up to sign). \\

When two classes cannot be distinguished by their high-level basis behavior, we need more refined invariants to separate them. For example, some invariants correspond to the specific action the gate has in some basis: the CNOT gate can generate any reversible linear transformation in the $Z$-basis; there is another gate which can only perform \emph{orthogonal} linear transformations; and we also have gates like $\CSIGN$, which can only change the sign of the basis element.\footnote{The CNOT (controlled-$X$) gate is sometimes written as CX, and we similarly write the controlled-$Z$ gate as $\CSIGN$.}  In fact, these three cases and the available collection of single-qubit gates are enough to determine the class.

% Sometimes the set of single-qubit operations generated by a set of gates serves to easily distinguish two classes. However, consider the $\frac{\pi}{4}$-phase gate and Pauli operations which are common to several classes. Even when combined with these single-qubit gates, CNOT can correlate bits in the $Z$-basis, whereas $\CSIGN$ cannot. In fact, there is a third class which can correlate bits in the $Z$-basis, but only by orthogonal transformations, separating it from both $\CSIGN$ and CNOT.

\item {\bf Finite Generation.} Every class can be generated by a single gate on at most four qubits. Also, given a set of gates generating some class, there always exist three gates from that set that generate the same class.  Moreover, the classification implies that the canonical set of Clifford generators---CNOT, Hadamard, and phase---is \emph{not} a minimal set of generators in our model.  It turns out that with the aid of ancilla qubits, CNOT and Hadamard generate a phase gate.  This is well-known \cite{aliferis:2007, jones:2012}, but comes as a simple consequence of our classification theorem.
\item {\bf Ancillas.} In general, giving a Clifford gate access to ancillary qubits often increases the set of functions it can compute.  A priori, one might suspect that extracting all functionality from a large entangling Clifford gate would require large highly-entangled ancilla states.  Nevertheless, our classification shows that only a constant number of one- and two-qubit ancillary states are ever needed.  In fact, an even stronger result is true.  Namely, our classification holds even when we allow the ancillas to change in an input-independent manner,\footnote{Aaronson et al.\ called this the ``Loose Ancilla Rule,'' and it \emph{does} affect their classification of classical reversible circuits.} as would be natural for a Clifford subroutine in a general quantum computation. See Section~\ref{sec:classes} for further discussion.
\item {\bf Canonical Forms.} It has long been known that there exists canonical forms for Clifford circuits \cite{ag, bravyi2021hadamard,maslov2018shorter, selinger2015generators}. Our classification theorem also reveals explicit canonical forms for most \emph{subclasses} of Clifford circuits (see Section~\ref{sec:equivalence}).  Furthermore, we give an explicit canonical form for 2-qubit Clifford circuits using the generators from our classification (see Appendix~\ref{app:canonical_form}).
\item {\bf Enumeration.}  For each class $\CC$ and for all $n$, we give explicit formulas for the number of gates in $\CC$ on $n$-qubits.  This enumeration is often derived from our explicit canonical forms discussed above.  One consequence is that every class is exponentially smaller than any class strictly containing it.  See Appendix~\ref{app:enumeration} for details.
\item {\bf Algorithms.}  Our classification implies a linear time algorithm which, given the tableau of a gate $G$, identifies which class $G$ belongs to. 
%As a consequence, we can identify the class for a list of tableaux (by running the algorithm for each one, then looking up the least upper bound on our lattice diagram), or determine whether one gate generates another (by running the algorithm for both, and checking whether the one class contains the other). 
In fact, to witness that $G$ generates some class in the classification, one only needs to view a constant number of bits of the tableau.  These details are discussed in sections \ref{sec:universal_construction} and \ref{sec:completing}.
\item {\bf Sporadic Gates.}  The process of classification unearthed certain strange classes, which arise from the interaction of the various invariants.  For example, four of the classes containing the $\tfour$ gate require a generator on at least four qubits.  Surprisingly, the class $\class{\tfour, \Gamma, \P}$ has a three-qubit generator.\footnote{Interestingly, there are no affine gate sets in the classification of \emph{classical} reversible gates which admit a generator over three bits and no smaller.}  We investigate such a gate in Appendix~\ref{app:gta}.  
\end{enumerate}}

%%%%%%% FIGURES %%%%%%%%%%
\begin{landscape}
\begin{figure}
% \begin{sidewaysfigure}

\begin{center}
\begin{tikzpicture}[>=latex,scale=0.25]
\tikzstyle{class}=[align=center, font=\small]
\tikzstyle{xpres}=[class]
\tikzstyle{ypres}=[class]
\tikzstyle{zpres}=[class] 
\tikzstyle{nonpres}=[class]
\tikzstyle{xyzpres}=[class]

\def\lbot{0}
\def\lvertex{10}
\def\lface{14}
\def\lpauli{24}
\def\ldih{34}
\def\la4{38}
\def\ltop{48}

\node[nonpres] (Bottom) at (0,\lbot) {$\bot$};
\node[nonpres] (Top) at (0,\ltop) {$\top$};
\node[xyzpres] (X) at (-8,\lface) {$X$};
\node[xyzpres] (Y) at ( 0,\lface) {$Y$};
\node[xyzpres] (Z) at ( 8,\lface) {$Z$};

\node[xpres] (Rx) at (-8,\lpauli) {$\R_X$};
\node[ypres] (Ry) at ( 0,\lpauli) {$\R_Y$};
\node[zpres] (Rz) at ( 8,\lpauli) {$\R_Z$};

\node[xpres] (P_Rx) at (-8,\ldih) {$\P + \R_X$};
\node[ypres] (P_Ry) at ( 0,\ldih) {$\P + \R_Y$};
\node[zpres] (P_Rz) at ( 8,\ldih) {$\P + \R_Z$};

\node[nonpres] (Gamma_mmp) at (-19,\lvertex) {$\Gamma_{--+}$};
\node[nonpres] (Gamma_mpm) at (-25,\lvertex) {$\Gamma_{-+-}$};
\node[nonpres] (Gamma_pmm) at (-31,\lvertex) {$\Gamma_{+--}$};
\node[nonpres] (Gamma_ppp) at (-37,\lvertex) {$\Gamma_{+++}$};

\node[xpres] (theta_ypz) at ( 18,\lvertex) {$\theta_{Y+Z}$};
\node[xpres] (theta_ymz) at ( 22,\lvertex) {$\theta_{Y-Z}$};
\node[ypres] (theta_xpz) at ( 26,\lvertex) {$\theta_{X+Z}$};
\node[ypres] (theta_xmz) at ( 30,\lvertex) {$\theta_{X-Z}$};
\node[zpres] (theta_xpy) at ( 34,\lvertex) {$\theta_{X+Y}$};
\node[zpres] (theta_xmy) at ( 38,\lvertex) {$\theta_{X-Y}$};

\node[nonpres] (P) at (-18,\lpauli) {$\P$};
\node[nonpres] (Gamma_P) at (-28,\la4) {$\P + \Gamma$};

\node[xpres] (X_theta_yz) at (20,\lpauli) {$X + \theta_{YZ}$};
\node[ypres] (Y_theta_xz) at (28,\lpauli) {$Y + \theta_{XZ}$};
\node[zpres] (Z_theta_xy) at (36,\lpauli) {$Z + \theta_{XY}$};

\node[nonpres] (Gamma_ppp_theta) at (19,\la4) {$\theta_{+++}$};
\node[nonpres] (Gamma_pmm_theta) at (25,\la4) {$\theta_{+--}$};
\node[nonpres] (Gamma_mpm_theta) at (31,\la4) {$\theta_{-+-}$};
\node[nonpres] (Gamma_mmp_theta) at (37,\la4) {$\theta_{--+}$};

\draw (Bottom) -- (X) -- (Rx) -- (P_Rx) -- (Top)
			(Bottom) -- (Y) -- (Ry) -- (P_Ry) -- (Top)
			(Bottom) -- (Z) -- (Rz) -- (P_Rz) -- (Top)
			(X) -- (P) -- (P_Rx)
			(Y) -- (P) -- (P_Ry)
			(Z) -- (P) -- (P_Rz)
			(X) -- (X_theta_yz) -- (P_Rx)
			(Y) -- (Y_theta_xz) -- (P_Ry)
			(Z) -- (Z_theta_xy) -- (P_Rz)
			(Bottom) -- (theta_ypz)
			(Bottom) -- (theta_ymz)
			(Bottom) -- (theta_xpz)
			(Bottom) -- (theta_xmz)
			(Bottom) -- (theta_xpy)
			(Bottom) -- (theta_xmy)
			(theta_ypz) -- (X_theta_yz)
			(theta_ymz) -- (X_theta_yz)
			(theta_xpz) -- (Y_theta_xz)
			(theta_xmz) -- (Y_theta_xz)
			(theta_xpy) -- (Z_theta_xy)
			(theta_xmy) -- (Z_theta_xy)
			(Bottom) -- (Gamma_ppp) -- (Gamma_P)
			(Bottom) -- (Gamma_pmm) -- (Gamma_P)
			(Bottom) -- (Gamma_mpm) -- (Gamma_P)
			(Bottom) -- (Gamma_mmp) -- (Gamma_P) -- (Top)
			(P) -- (Gamma_P)
			(Gamma_ppp_theta) -- (Top)
			(Gamma_pmm_theta) -- (Top)
			(Gamma_mpm_theta) -- (Top)
			(Gamma_mmp_theta) -- (Top);

\draw	[gray, dotted]		
			(theta_ypz) -- (Gamma_ppp_theta)
			(theta_xpz) -- (Gamma_ppp_theta)
			(theta_xpy) -- (Gamma_ppp_theta) -- (Top)
			(theta_ypz) -- (Gamma_pmm_theta)
			(theta_xmz) -- (Gamma_pmm_theta)
			(theta_xmy) -- (Gamma_pmm_theta) -- (Top)
			(theta_ymz) -- (Gamma_mpm_theta)
			(theta_xpz) -- (Gamma_mpm_theta)
			(theta_xmy) -- (Gamma_mpm_theta) -- (Top)
			(theta_ymz) -- (Gamma_mmp_theta)
			(theta_xmz) -- (Gamma_mmp_theta)
			(theta_xpy) -- (Gamma_mmp_theta) -- (Top)
			(Gamma_ppp) -- (Gamma_ppp_theta)
			(Gamma_pmm) -- (Gamma_pmm_theta)
			(Gamma_mpm) -- (Gamma_mpm_theta)
			(Gamma_mmp) -- (Gamma_mmp_theta);
\end{tikzpicture}
\end{center}
\caption{The inclusion lattice of degenerate Clifford gate classes. To improve readability, some lines are gray and dotted, but they have the same meaning as their counterparts. The gates in this diagram are listed in Section~\ref{sec:single_qubit_gates} and some class names have been abbreviated as listed in Section~\ref{sec:classes}.}
\label{fig:degeneratelattice}
% \end{sidewaysfigure}
\end{figure}
\end{landscape}

\input{figures/multi_qubit_lattice}

\subsection{Proof Outline}
We can divide the proof into a few major steps.   First, we introduce the notion of a tableau, a binary matrix representation of a Clifford circuit.  We then present all the classes in the classification and designate them by their generators.  An examination of the tableaux of the gates in these classes reveals candidate invariants.  We then prove that these candidate invariants are indeed invariant under the circuit building operations in our model.  That is to say, if we have two gates whose tableaux satisfy the invariant, then the tableau of their composition satisfies the invariant, and so on for all the other ways to build circuits from gates---tensor products, ancillas, swapping.  

At this point, we will have shown that each class has a corresponding invariant, which implies that each class in our lattice is distinct. That is, for any two classes, there is a generator of one that fails to satisfy the invariant of the other.  Next, we will show that this correspondence is complete.  The generators of a class can construct \emph{any} gate which satisfies the invariant for that class.

The challenge remains to show that our list of classes is exhaustive.  Suppose we are given some gate set $G$, and we wish to identify the class it generates.  Clearly, the class generated by $G$ is contained in some class in the lattice, and let $\CC$ be the smallest such class. The hope is to show that $G$ generates all of $\CC$. To do this, we use the minimality of $\CC$. That is, for each class $\mathcal{S} \subset \CC$, there must be some gate $g \in G$ which is not in $\mathcal{S}$, otherwise $\mathcal{S}$ would be a smaller class containing $G$. We now wish to use $g$ to generate a simpler gate, also violating some invariant of $\mathcal{S}$. This is accomplished via the ``universal construction,'' which is a particular circuit built from $g$ and SWAP gates. Finally, we combine the simpler gates to construct the canonical generators for the class $\CC$ itself.

%!TEX root = ../full_paper.tex

\section{Stabilizer Formalism}
\label{sec:stabilizer_formalism}
The one-qubit unitary operations 
\begin{align*}
X &= \begin{pmatrix} 0 & 1 \\ 1 & 0 \end{pmatrix} & Y &= \begin{pmatrix} 0 & -i \\ i & 0 \end{pmatrix} & Z &= \begin{pmatrix} 1 & 0 \\ 0 & -1 \end{pmatrix}
\end{align*}
are known as \emph{Pauli matrices}. The Pauli matrices are all involutions ($X^2 = Y^2 = Z^2 = I$), and have the following relations between them
\begin{align*}
XY &= iZ & YZ &= iX & ZX &= iY \\
YX &= -iZ & ZY &= -iX & XZ &= -iY.
\end{align*}
It follows that the Pauli matrices generate a discrete group (under multiplication), called the \emph{Pauli group} $\P$, which consists of sixteen elements: $\{ I, X, Y, Z \}$ with phases $\pm 1$ or $\pm i$. The \emph{Pauli group on $n$ qubits}, $\P_n$, is the set of all $n$-qubit tensor products of elements from $\P$. We define a \emph{Pauli string} as any element of $\P_n$ with positive phase (i.e., a tensor product of the matrices $I, X, Y, Z$). We frequently omit the tensor product symbol from Pauli strings and write, e.g., $P_1 \cdots P_n$ where we mean $P_1 \otimes \cdots \otimes P_n$. 

The \emph{Clifford group on $n$ qubits}, $\clifford_n$, is the set of unitary operations which \emph{normalize} $\P_n$ in the group-theoretic sense. That is, $U \in \clifford_n$ if $U p U^{\dagger} \in \P_n$ for all $p \in \P_n$. We leave it as a simple exercise to check that $\clifford_n$ is indeed a group. 

A \emph{Clifford gate} is any unitary in $\bigcup_{n \geq 1} \clifford_n$. A \emph{Clifford circuit} is a quantum circuit of Clifford gates implementing a unitary transformation on some set of qubits, designated the \emph{input/output qubits}, while preserving the state of the remaining \emph{ancilla qubits}.  We say that a state $\ket{\psi}$ is \emph{stabilized} by an operation $U$ if and only if $U\ket{\psi} = \ket{\psi}$. In other words, $\ket{\psi}$ is in the $+1$ eigenspace of $U$.  The Pauli elements and their corresponding stabilized states are below:

\begin{center}
\begin{tabular}{l l}
${X} : \ket{+} = \frac{\ket{0} + \ket{1}}{\sqrt{2}}$ & $-{X} : \ket{-} = \frac{\ket{0} - \ket{1}}{\sqrt{2}}$ \\
${Y} : \ket{i} = \frac{\ket{0} + i \ket{1}}{\sqrt{2}}$ & $-{Y} : \ket{-i} = \frac{\ket{0} - i \ket{1}}{\sqrt{2}}$ \\
${Z} : \ket{0}$ & $-{Z} : \ket{1}$ 
\end{tabular}
\end{center}
We call the vectors stabilized by non-identity Pauli elements $P$ and $-P$ the \emph{$P$-basis}. A \emph{stabilizer state} is any state $U \ket{0 \ldots 0}$ where $U$ is a Clifford gate.  For example, $\ket{0}$, $\ket{1}$, $\ket{+}$, $\ket{-}$, $\ket{i}$, and $\ket{-i}$ are the 6 stabilizer states on one qubit. Multi-qubit stabilizer states include $\frac{1}{\sqrt{2}} (\ket{0^n} + \ket{1^n})$ and $\sum_{x \in \{ 0, 1 \}^{n}} \ket{x}$.  In general, stabilizer states are of the form (unnormalized) $\sum_{x \in A} (-1)^{q(x)}i^{\ell(x)}\ket{x}$ where $A$ is an affine space over $\ftwo$, $q(x)$ is a quadratic form, and $\ell(x)$ is a linear form \cite{dehaene:2003,nest:2008}.

%A stabilizer state $\ket{\psi}$ on $n$ qubits can be equivalently specified by $n$ linearly independent elements from $\P_n$ such that each element stabilizes $\ket{\psi}$.  Because of this equivalence between states and the list of operations that stabilize them, one can track of the state of a system by simply updating a list of stabilizing operations.   More precisely, suppose we have a state $\ket{\psi}$ stabilized by operation $P$.  Suppose now that we apply some unitary $U$ to $\ket{\psi}$.  Let $P' = U P U^\dag$ be the conjugation of $P$ by $U$.  We now have $P' (U \ket{\psi}) = U P U^\dag U \ket{\psi} = U P \ket{\psi} = U \ket{\psi}$.  Since $P'$ stabilizes $U \ket{\psi}$, by maintaining the list of stabilizers under conjugation we can also recover the underlying state of our system.  This reasoning is the foundation for the now-famous Gottesman-Knill Theorem \cite{gottesman:1998}.

%!TEX root = ../full_paper.tex

\section{Gates}
\label{sec:gates}

Let us introduce some common Clifford gates used throughout the classification. 

\subsection{Single-qubit Gates}
\label{sec:single_qubit_gates}

We start with the single-qubit Clifford gates, which by definition permute (up to phases) the $X$, $Y$, and $Z$ bases.  In fact, the single-qubit Clifford gates correspond to symmetries of the cube (see Figure~\ref{fig:cube}).\footnote{Or equivalently, symmetries of the octahedron, which is dual to the cube.} We group the gates by the type of rotation to emphasize this geometric intuition.

\begin{description}
\item[Face rotations:] The Pauli matrices $X$, $Y$, and $Z$ (as gates) correspond to $180^{\circ}$ rotations about the $X$, $Y$, and $Z$ axes respectively. Similarly, we define $\R_X$, $\R_Y$, and $\R_Z$ to be $90^{\circ}$ rotations (in the counterclockwise direction) about their respective axes. Formally, 
\begin{align*}
\R_X &= \frac{I - iX}{\sqrt{2}}, & \R_Y &= \frac{I - iY}{\sqrt{2}}, & \R_Z &= \frac{I - iZ}{\sqrt{2}},
\end{align*}
although in the case of $\R_Z$ (also known as the \emph{phase gate} and often denoted by $S$ or $P$), a different choice of phase is more conventional. The clockwise rotations are then $\R_X^{\dagger}$, $\R_Y^{\dagger}$, and $\R_Z^{\dagger}$. 
\item[Edge rotations:] Another symmetry of the cube is to rotate one of the edges $180^{\circ}$. Opposing edges produce the same rotation, so we have six gates: $\theta_{X+Y}$, $\theta_{X-Y}$, $\theta_{X+Z}$, $\theta_{X-Z}$, $\theta_{Y+Z}$, $\theta_{Y-Z}$. We define 
\begin{align*}
\theta_{P+Q} &= \frac{P+Q}{\sqrt{2}}, & \theta_{P-Q} &= \frac{P-Q}{\sqrt{2}},
\end{align*}
for all Pauli matrices $P \neq Q$. Note that $\theta_{X+Z}$ is the well-known \emph{Hadamard gate}, usually denoted by $H$.
\item[Vertex rotations:] The final symmetry is a $120^{\circ}$ counterclockwise rotation around one of the diagonals passing through opposite vertices of the cube. The cube has eight vertices, $(\pm 1, \pm 1, \pm 1)$, and we denote the corresponding single-qubit gates $\Gamma_{+++}$, $\Gamma_{++-}$, $\ldots$, $\Gamma_{---}$. Algebraically, we define 
\begin{align*}
\Gamma_{+++} &= \frac{I - iX - iY - iZ}{2}, \\
\Gamma_{++-} &= \frac{I - iX - iY + iZ}{2}, \\
&\vdots \\
\Gamma_{---} &= \frac{I + iX + iY + iZ}{2}.
\end{align*}
We also define $\Gamma$ (without subscripts) to be the first gate, $\Gamma_{+++}$, since it is the most convenient; conjugation by $\Gamma$ maps $X$ to $Y$, $Y$ to $Z$, and $Z$ to $X$.
\end{description}

\tdplotsetmaincoords{60}{110}
\begin{figure}
\begin{center}
\begin{tikzpicture}[scale=1.4,tdplot_main_coords,>=stealth,show background rectangle]

\draw[gray] (-2,-2, 2) -- (-2,-2,-2);
\draw[gray] ( 2,-2,-2) -- (-2,-2,-2);
\draw[gray] (-2, 2,-2) -- (-2,-2,-2);
\draw[thick] ( 2, 2, 2) -- (-2, 2, 2);
\draw[thick] ( 2, 2, 2) -- ( 2,-2, 2);
\draw[thick] ( 2, 2, 2) -- ( 2, 2,-2);
\draw[thick] (-2, 2, 2) -- (-2,-2, 2);
\draw[thick] (-2, 2, 2) -- (-2, 2,-2);
\draw[thick] ( 2,-2, 2) -- (-2,-2, 2);
\draw[thick] ( 2,-2, 2) -- ( 2,-2,-2);
\draw[thick] ( 2, 2,-2) -- (-2, 2,-2);
\draw[thick] ( 2, 2,-2) -- ( 2,-2,-2);

\draw[very thick,->,dashed] (2,0,0) -- (5,0,0) node[right] {$\R_X$};
\draw[very thick,->,dashed] (0,2,0) -- (0,5,0) node[above] {$\R_Y$};
\draw[very thick,->,dashed] (0,0,2) -- (0,0,5) node[left] {$\R_Z$};

\draw[->] ( 2, 2, 2) -- ( 3, 3, 3) node[right] {$\Gamma_{+++}$};
\draw[->] ( 2,-2, 2) -- ( 3,-3, 3) node[left] {$\Gamma_{+-+}$};
\draw[->] (-2, 2, 2) -- (-3, 3, 3) node[above] {$\Gamma_{-++}$};
\draw[->] (-2,-2, 2) -- (-3,-3, 3) node[above] {$\Gamma_{--+}$};
\draw[->] ( 2, 2,-2) -- ( 3, 3,-3) node[below] {$\Gamma_{++-}$};
\draw[->] ( 2,-2,-2) -- ( 3,-3,-3) node[below] {$\Gamma_{+--}$};
\draw[->] (-2, 2,-2) -- (-3, 3,-3) node[right] {$\Gamma_{-+-}$};
\draw[->] (-2,-2,-2) -- (-3,-3,-3) node[above] {$\Gamma_{---}$};

\draw[->,dotted] ( 2, 2, 0) -- ( 3, 3, 0) node[right] {$\theta_{X+Y}$};
\draw[->,dotted] ( 2,-2, 0) -- ( 3,-3, 0) node[left] {$\theta_{X-Y}$};
\draw[->,dotted] ( 2, 0, 2) -- ( 3, 0, 3) node[above right] {$\theta_{X+Z}$};
\draw[->,dotted] ( 2, 0,-2) -- ( 3, 0,-3) node[below] {$\theta_{X-Z}$};
\draw[->,dotted] ( 0, 2, 2) -- ( 0, 3, 3) node[above right] {$\theta_{Y+Z}$};
\draw[->,dotted] ( 0, 2,-2) -- ( 0, 3,-3) node[below] {$\theta_{Y-Z}$};

\end{tikzpicture}
\end{center}
\caption{Single-qubit gates as symmetries of the cube.}
\label{fig:cube}
\end{figure}

\begin{table}
\centering
\newcommand{\tablespacer}{10pt}
\begin{tabular}{c | c | c}
Gate & Tableau & Unitary Matrix \\ \hline
$X$ & $\left(\begin{array}{cc|c} 1 & 0 &0\\  0 & 1 &1\end{array}\right)$ & $\left(\begin{array}{cc} 0 & 1 \\ 1 & 0 \end{array}\right)$  \\ [\tablespacer]
$Y$ & $\left(\begin{array}{cc|c} 1 & 0 &1\\  0 & 1 &1\end{array}\right)$ & $\left(\begin{array}{cc} 0 & -i \\ i & 0 \end{array}\right)$  \\ [\tablespacer]
$Z$ & $\left(\begin{array}{cc|c} 1 & 0 &1\\  0 & 1 &0\end{array}\right)$ & $\left(\begin{array}{cc} 1 & 0 \\ 0 & -1 \end{array}\right)$  \\ [\tablespacer]
$\R_X$ & $\left(\begin{array}{cc|c} 1 & 0 &0\\  1 & 1 & 1\end{array}\right)$ & $\frac{1}{\sqrt{2}} \left(\begin{array}{cc} 1 & -i \\ -i & 1 \\\end{array}\right)$  \\ [\tablespacer]
$\R_Y$ & $\left(\begin{array}{cc|c} 0 & 1 &1\\  1 & 0 & 0\end{array}\right)$ & $\frac{1}{\sqrt{2}} \left(\begin{array}{cc} 1 & -1 \\ 1 & 1 \\\end{array}\right)$  \\ [\tablespacer]
$\R_Z = S$ & $\left(\begin{array}{cc|c} 1 & 1 &0\\  0 & 1 &0\end{array}\right)$ & $\frac{1-i}{\sqrt{2}} \left(\begin{array}{cc} 1 & 0 \\ 0 & i \end{array}\right)$  \\ [\tablespacer]
$\theta_{XZ} = \theta_{X+Z} = H$ & $\left(\begin{array}{cc|c} 0 & 1 &0\\  1 & 0 &0\end{array}\right)$ & $\frac{1}{\sqrt{2}}\left(\begin{array}{cc} 1 & 1 \\ 1 & -1 \end{array}\right)$  \\ [\tablespacer]
$\theta_{X-Z}$ & $\left(\begin{array}{cc|c} 0 & 1 &1\\  1 & 0 &1\end{array}\right)$ & $\frac{1}{\sqrt{2}}\left(\begin{array}{cc} -1 & 1 \\ 1 & 1 \end{array}\right)$  \\ [\tablespacer]
$\Gamma_{+++} = \Gamma$ & $\left(\begin{array}{cc|c} 1 & 1 &0 \\  1 & 0 &0\end{array}\right)$ & $\frac{1-i}{2}\left(\begin{array}{cc} 1 & -i \\ 1 & i \\\end{array}\right)$  \\ [\tablespacer]
\end{tabular}
\caption{Single-qubit gates represented as tableaux (Section~\ref{sec:tableaux}) and as complex unitary matrices.}
\label{table:single_qubit_gates}
\end{table}

\subsection{Multi-qubit Gates}
\label{sec:multi_gates}

We now introduce the multi-qubit Clifford gates relevant to the classification.\footnote{Like the single-qubit gates, it turns out that the two-qubit Clifford gates can also be interpreted as symmetries of a polyhedron, in particular, the six-dimensional hyperoctahedron \cite{glaudell2021optimal, wiki:Hyperoctahedral_group}.  However, we find it easier to reason about the group from a set of simple generators.} The \emph{SWAP gate}, for instance, simply exchanges two qubits. A more interesting example is the \emph{controlled-NOT} or \emph{CNOT gate}, and the \emph{generalized CNOT gates}. 

A \emph{generalized CNOT gate} is a two-qubit Clifford gate of the form 
$$
C(P, Q) := \frac{I \otimes I + P \otimes I + I \otimes Q - P \otimes Q}{2},
$$
where $P$ and $Q$ are Pauli matrices. If the first qubit is in the $+1$ eigenspace of $P$ then $C(P,Q)$ does nothing, but if it is in the $-1$ eigenspace of $P$ then $C(P,Q)$ applies $Q$ to the second qubit. Of course, the definition is completely symmetric, so you can also view it as applying $P$ to the first qubit when the second qubit is in the $-1$ eigenspace of $Q$. 

Observe that $C(Z, X)$ is actually the $\operatorname{CNOT}$ gate; it applies a $\operatorname{NOT}$ gate to the second qubit when the first qubit is $\ket{1}$ and does nothing when the first qubit is $\ket{0}$. Figure~\ref{fig:cnot_as_czx} shows this equivalence, and illustrates our circuit diagram notation for generalized CNOT gates. Also note that $C(X, Z)$ is a $\operatorname{CNOT}$, but with the opposite orientation (i.e., the second bit controls the first). The rest of the \emph{heterogeneous generalized CNOT gates} (i.e., $C(P,Q)$ where $P \neq Q$) are the natural equivalents of $\operatorname{CNOT}$ in different bases. 

\begin{figure}
\begin{align*}
\Qcircuit @C=2em @R=1.5em {
& \gate{Z}  & \qw  \\
& \gate{X} \qwx & \qw}
\horizontally
\raisebox{-15pt}{=}
\horizontally
\Qcircuit @C=2em @R=2.3em {
& \ctrl{1} & \qw \\
& \targ & \qw}
\end{align*}
\caption{$\CNOT$ expressed as a $C(Z,X)$ gate.}
\label{fig:cnot_as_czx}
\end{figure}

Similarly, $C(Z,Z)$ sometimes known as the \emph{controlled-sign gate} or $\CSIGN$, which flips the sign on input $\ket{11}$, but does nothing otherwise. The \emph{homogeneous generalized CNOT gates} (i.e., $C(P, P)$ for some $P$) are quite different from heterogeneous CNOT gates. For instance, when one $\operatorname{CNOT}$ targets the control qubit of another $\operatorname{CNOT}$ then it matters which gate is applied first. On the other hand, two $\CSIGN$ gates will always commute, whether or not they have a qubit in common. 

It turns out that every two-qubit Clifford gate is equivalent (up to a SWAP) to some combination of one qubit Clifford gates and at most one generalized CNOT gate (see Appendix~\ref{app:canonical_form}). Although most classes of Clifford gates can be specified by such two-qubit generators, there are five classes which require a larger generator such as the $\tfour$ gate \cite{ags:2015, gottesman:1998theory}.  

For all $k \geq 1$, let $\tgate_{2k}$ be a $2k$-qubit gate such that for all $x =(x_1, \ldots, x_{2k}) \in \{0,1\}^{2k}$
$$\tgate_{2k}\ket{x_1, \ldots, x_{2k}} = \ket{x_1 \oplus b_x, x_2 \oplus b_x, \ldots, x_{2k} \oplus b_x}$$
where $b_x = x_1 \oplus x_2 \oplus \ldots \oplus x_{2k}$. Intuitively, $\tgate_{2k}$ outputs the complement of the input when the parity of the input is odd and does nothing when the parity of the input is even.  In particular, $\tgate_2$ is the lowly SWAP gate.   Notice that this is an orthogonal linear function of the input bits, and therefore has a $2k \times 2k$ matrix over $\ftwo$ with orthogonal rows and columns:
%which hints at an invariant which may arise in the classification.

$$
\tgate_{2k} =
\begin{pmatrix}
0 & 1 & 1 & \cdots & 1 \\
1 & 0 & 1 & \cdots & 1 \\
1 & 1 & 0 & \cdots & 1 \\
\vdots & \vdots & \vdots & \ddots & \vdots \\
1 & 1 & 1 & \cdots & 0 \end{pmatrix}
$$

% The matrix of the $\tfour$ gate (over $\ftwo$) is
%$$\begin{pmatrix}
%0 & 1 & 1 & 1 \\
%1 & 0 & 1 & 1 \\
%1 & 1 & 0 & 1 \\
%1 & 1 & 1 & 0
%\end{pmatrix}.$$
%In the quantum setting, let the $\operatorname{T}_{2k}$ gate simply apply this transformation to the $Z$ basis vectors.  
%In fact, we will see later that $\operatorname{T}_{2k}$ applies the same transformation to the $X$-basis if we identify $\ket{0}$ with $\ket{+}$ and $\ket{1}$ with $\ket{-}$.  

\begin{table}
\centering
\newcommand{\tablespacer}{25pt}
\begin{tabular}{c | c | c}
Gate & Tableau & Unitary Matrix on computational basis \\ \hline
$C(X,X)$ & $\left(\begin{array}{cc|cc} 1 & 0 & 0 & 0 \\ 0 & 1 & 1 & 0 \\ \hline 0 & 0 & 1 & 0 \\ 1 & 0 & 0 & 1 \\\end{array}\right)$  & 
$\frac{1}{2}\left(\begin{array}{cccc} 1 & 1 & 1 & -1 \\ 1 & 1 & -1 & 1 \\ 1 & -1 & 1 & 1 \\ -1 & 1 & 1 & 1 \end{array}\right)$ \\ [\tablespacer]
$C(Y,Y)$ & $\left(\begin{array}{cc|cc} 1 & 0 & 1 & 1 \\ 0 & 1 & 1 & 1 \\ \hline 1 & 1 & 1 & 0 \\ 1 & 1 & 0 & 1 \\\end{array}\right)$ & 
$\frac{1}{2}\left(\begin{array}{cccc} 1 & -i & -i & 1 \\ i & 1 & -1 & -i \\ i & -1 & 1 & -i \\ 1 & i & i & 1 \\\end{array}\right)$ \\ [\tablespacer]
$C(Z,Z)$ & $\left(\begin{array}{cc|cc} 1 & 0 & 0 & 1 \\ 0 & 1 & 0 & 0 \\ \hline 0 & 1 & 1 & 0 \\ 0 & 0 & 0 & 1 \\\end{array}\right)$ & 
$\left(\begin{array}{cccc} 1 & 0 & 0 & 0 \\ 0 & 1 & 0 & 0 \\ 0 & 0 & 1 & 0 \\ 0 & 0 & 0 & -1 \\\end{array}\right)$\\ [\tablespacer]
$C(Y,X)$ & $\left(\begin{array}{cc|cc} 1 & 0 & 1 & 0 \\ 0 & 1 & 1 & 0 \\ \hline 0 & 0 & 1 & 0 \\ 1 & 1 & 0 & 1 \\\end{array}\right)$ & 
$\frac{1}{2} \left(\begin{array}{cccc} 1 & 1 & -i & i \\ 1 & 1 & i & -i \\ i & -i & 1 & 1 \\ -i & i & 1 & 1 \\\end{array}\right)$\\ [\tablespacer]
$C(Z,Y)$ & $\left(\begin{array}{cc|cc} 1 & 0 & 1 & 1 \\ 0 & 1 & 0 & 0 \\ \hline 0 & 1 & 1 & 0 \\ 0 & 1 & 0 & 1 \\\end{array}\right)$ & 
$\left(\begin{array}{cccc} 1 & 0 & 0 & 0 \\ 0 & 1 & 0 & 0 \\ 0 & 0 & 0 & -i \\ 0 & 0 & i & 0 \\\end{array}\right)$\\ [\tablespacer]
$C(X,Z) = \CNOT$ & $\left(\begin{array}{cc|cc} 1 & 0 & 0 & 0 \\ 0 & 1 & 0 & 1 \\ \hline 1 & 0 & 1 & 0 \\ 0 & 0 & 0 & 1 \\\end{array}\right)$ & 
$\left(\begin{array}{cccc} 1 & 0 & 0 & 0 \\ 0 & 0 & 0 & 1 \\ 0 & 0 & 1 & 0 \\ 0 & 1 & 0 & 0 \\\end{array}\right)$\\ [\tablespacer]
\end{tabular}
\caption{Two qubit gates.  The sign bits are all 0 in the above tableaux so they are omitted.}
\label{table:2_qubit_gates}
\end{table}

%!TEX root = ../full_paper.tex

\section{Tableaux}
\label{sec:tableaux}

Observe that the matrices $I, X, Y, Z$ are linearly independent, and therefore form a basis for the $2 \times 2$ complex matrices. It follows that $\P_n$ spans all $2^n \times 2^n$ complex matrices. Hence, any unitary operation on $n$ qubits can be characterized by its action on $\P_n$. In particular, any gate is characterized by how it acts on 
$$
p_1 = XI \cdots I, \; p_2 = ZI \cdots I, \; \ldots, \;p_{2n-1} = I\cdots IX, \; p_{2n} = I\cdots IZ.
$$
We call this list the \emph{Pauli basis on $n$ qubits}, since one can write any element of $\P_n$ as a product of basis elements times a phase ($\pm 1$ or $\pm i$).

Now suppose we are given a Clifford gate, $U \in \clifford_n$. By definition, Clifford gates map each Pauli basis element to something in $\P_n$, which can be written as a product of basis elements times a phase. That is, 
$$
U p_j U^{\dagger} = \alpha_j \prod_{k=1}^{2n} p_k^{M_{jk}}
$$
for some bits $M_{j1}, \ldots, M_{j(2n)} \in \{ 0, 1 \}$ and some phase\footnote{Note that the order of the terms in the product matters, since the Pauli basis elements do not necessarily commute, so we assume the terms are in the natural order from $p_1^{M_{j1}}$ up to $p_{2n}^{M_{j(2n)}}$.} $\alpha_j \in \{ \pm 1, \pm i \}$. The \emph{tableau} for $U$ is a succinct representation for $U$ consisting of the binary matrix $M = [M_{jk}]$, and some representation of the phases $\alpha_1, \ldots, \alpha_{2n}$. 

It turns out that $U$ maps $p_j$ (or any Pauli string) to $\pm 1$ times a Pauli string. This follows from the fact that each Pauli string has $\pm 1$ eigenvalues, which are preserved under conjugation by a unitary. However, $\alpha_j$ may still be any one of $\{ \pm 1, \pm i \}$. This is because the product of $p_{2k-1} p_{2k}$ is $I \cdots I (XZ) I \cdots I = -i I \cdots IYI \cdots I$, with an awkward $-i$ phase. Once we cancel the extra factors of $i$ from $\alpha_j$, we are left with
$$
(-1)^{v_j} := \alpha_j \prod_{k=1}^{n} (-i)^{M_{j(2k-1)} M_{j(2k)}},
$$
where $v_j \in \{ 0, 1 \}$ is the \emph{phase bit for row $j$}. For example, if $U p_1 U^{\dagger}$ is $YI \cdots I$ then we have $M_{11} = M_{12} = 1$ and $v_1 = 0$. Then the complete tableau for $U$ is the matrix $M = [M_{jk}]$ and the vector of bits $v = [v_j]$, which we typically write as 
$$
\left(
\begin{array}{ccc|c}
M_{11} & \cdots & M_{1(2n)} & v_1 \\
\vdots & \ddots & \vdots & \vdots \\
M_{(2n)1} & \cdots & M_{(2n)(2n)} & v_{2n}
\end{array}
\right).
$$

Our ordering of the basis elements (which differs from other presentations \cite{ag}) puts Pauli strings on the same qubit (e.g., $XI\cdots I$ and $ZI \cdots I$) side-by-side in the matrix, so the $2 \times 2$ submatrix 
$$
\begin{pmatrix}
M_{2i-1,2j-1} & M_{2i-1,2j} \\
M_{2i,2j-1} & M_{2i,2j}
\end{pmatrix}
$$
completely describes how the $i$th qubit of the input affects the $j$th qubit of the output. In fact, it will be fruitful to think of the tableau as an $n \times n$ matrix of $2 \times 2$ blocks, along with a vector of $2n$ \emph{phase bits}. To be clear, the blocks come from $\ring := \ftwo^{2 \times 2}$, the ring of $2 \times 2$ matrices over the field of two elements, $\ftwo$. Then the tableau is a matrix in $\ring^{n \times n}$ (the $n \times n$ matrices over the ring $\ring$), combined with a vector of phase bits in $\ftwo^{2n}$. Each row of the matrix is associated with a pair of phase bits from the vector. 

However, not every matrix in  $\ring^{n \times n}$ corresponds to a unitary operation and therefore to a Clifford gate. To help express valid tableau, we define a unary operation $^*$ on $\ring$ such that
$$
\begin{pmatrix}
a & b \\
c & d
\end{pmatrix}^{*} 
= 
\begin{pmatrix}
d & b \\
c & a
\end{pmatrix}.
$$
The $^*$ operator has the property that
$$
\begin{pmatrix}
a & b \\
c & d
\end{pmatrix}
\begin{pmatrix}
a & b \\
c & d
\end{pmatrix}^{*} 
= 
\begin{pmatrix}
a & b \\
c & d
\end{pmatrix}
\begin{pmatrix}
d & b \\
c & a
\end{pmatrix}
=
\begin{pmatrix}
ad+bc & 0 \\
0 & ad+bc
\end{pmatrix}
=
I \begin{vmatrix}
a & b \\
c & d
\end{vmatrix}.
$$
Additionally, 
\begin{align*}
I^* &= I, \\
(M+N)^* &= M^* + N^*, \\
(MN)^* &= N^* M^*, \\
(M^*)^* &= M,
\end{align*}
so $^*$ makes $\ring$ a \emph{$^*$-ring} or \emph{ring with involution}. We also extend $^*$ to an operation on matrices (over $\ring$) which applies $^*$ to each entry and then transposes the matrix. It turns out that a tableau represents a unitary operation if and only if the matrix $M \in \ring^{n \times n}$ satisfies $MM^{*} = M^* M = I$. This (intentionally) resembles the definition of a unitary matrix ($U U^{\dagger} = U^{\dagger} U = I$), but we will call this the \emph{symplectic condition}.  This follows the traditional presentation of the tableau as a symplectic matrix over $\ftwo$. 
%and it corresponds to the unitarity of $U$ as a gate, but $M$ is certainly not a traditional complex unitary matrix (nor unitary over some finite field with conjugation). 

\subsection{Correspondence between Gates and Tableaux}
We will find it useful to switch between gates and tableaux, as one notion often informs the other.  Since the phase bits will often be irrelevant, let us denote the matrix part of the tableau for $g$ as $\Tstar(g)$. Most non-degenerate gate sets generate the Pauli group, which alone suffices to set the phase bits of the tableau arbitrarily by applying gates at the beginning of the circuit as follows: applying $X$ to qubit $j$ negates $v_{2j}$ and applying $Z$ to qubit $j$ negates $v_{2j-1}$. Furthermore, there is a surprising connection between individual entries of tableaux and elementary Clifford operations that can be extracted from them (see Section~\ref{sec:universal_construction}).

If $a \in \ring$ is invertible, then let $\G(a)$ be the single-qubit gate with $\Tstar(\G(a)) = a$ and zeros for phase bits.  These gates are well-defined since every $2 \times 2$ invertible matrix over $\mathbb F_2$ is itself a symplectic matrix over $\mathbb F_2$.  They are shown in the first row of Table~\ref{table:invert}.  Let $\G(a,i)$ be the gate $\G(a)$ applied to the $i$th qubit. 
% $\mathrm{GL}(2, \mathbb F_2) = \mathrm{Sp}(2, \mathbb F_2)$

\begin{table}
\centering
\setlength\extrarowheight{2.5pt}
\begin{tabular}{| c | c | c | c | c | c | c |} \hline
& $(\begin{smallmatrix} 1 & 0 \\ 0 & 1 \end{smallmatrix})$ & 
$(\begin{smallmatrix} 1 & 0 \\ 1 & 1 \end{smallmatrix})$ & 
$(\begin{smallmatrix} 0 & 1 \\ 1 & 0 \end{smallmatrix})$ &
$(\begin{smallmatrix} 1 & 1 \\ 0 & 1 \end{smallmatrix})$ & 
$(\begin{smallmatrix} 1 & 1 \\ 1 & 0 \end{smallmatrix})$ &
$(\begin{smallmatrix} 0 & 1 \\ 1 & 1 \end{smallmatrix})$ \\[2.5pt] \hline
$(\begin{smallmatrix} 0  \\ 0 \end{smallmatrix})$ & $I$ & $\R_X^\dag$ & $\theta_{X+Z}$ & $\R_Z$ & $\Gamma_{+++}$ & $\Gamma_{---}$\\[2.5pt] \hline
$(\begin{smallmatrix} 0  \\ 1 \end{smallmatrix})$ & $X$ & $\R_X$ & $\R_Y^\dag$ & $\theta_{X+Y}$ &  $\Gamma_{--+}$ & $\Gamma_{+-+}$\\[2.5pt] \hline
$(\begin{smallmatrix} 1  \\ 0 \end{smallmatrix})$ & $Z$ & $\theta_{Y+Z}$ & $\R_Y$ & $\R_Z^\dag$ & $\Gamma_{-+-}$ & $\Gamma_{-++}$\\[2.5pt] \hline
$(\begin{smallmatrix} 1  \\ 1 \end{smallmatrix})$ & $Y$ & $\theta_{Y-Z}$ & $\theta_{X-Z}$ & $\theta_{X-Y}$ & $\Gamma_{+--}$ & $\Gamma_{++-}$\\[2.5pt] \hline
\end{tabular}
\caption{Invertible tableau elements and the corresponding single-qubit gates produced by the universal construction.  Row of the table corresponds to the sign bit of the row of the tableau in which the element occurs.}
\label{table:invert}
\end{table}

We also define gates from the noninvertible elements of $\ring$.  We see that the tableau of each generalized CNOT gate consists of unique noninvertible elements $b \in \ring$ and $b^* \in \ring$ along the off-diagonal, as shown in Table~\ref{table:2_qubit_gates}.  We summarize this correspondence more succinctly in Table~\ref{table:noninvert}. Therefore, if $b \in \ring$ is noninvertible, define $\CNOT(b, i, j)$ to be the generalized CNOT on qubits $i$ and $j$ corresponding to the noninvertible matrix $b$. The tableau for $\CNOT(b, i, j)$ is the identity tableau except for $b^*$ and $b$ in positions $(i,j)$ and $(j,i)$, respectively. We use the circuit in Figure~\ref{fig:b_CNOT} to designate such a gate.

\begin{table}[ht!]
\centering
\setlength\extrarowheight{2.5pt}
\begin{tabular}{| c | c | c | c | c | c | c |} \hline
Element &  
$(\begin{smallmatrix} 0 & 0 \\ 1 & 0 \end{smallmatrix})$ &
$(\begin{smallmatrix} 1 & 1 \\ 1 & 1 \end{smallmatrix})$ & 
$(\begin{smallmatrix} 0 & 1 \\ 0 & 0 \end{smallmatrix})$ & 
$(\begin{smallmatrix} 1 & 0 \\ 0 & 0 \end{smallmatrix})$ / $(\begin{smallmatrix} 0 & 0 \\ 0 & 1 \end{smallmatrix})$ & 
$(\begin{smallmatrix} 1 & 0 \\ 1 & 0 \end{smallmatrix})$ / $(\begin{smallmatrix} 0 & 0 \\ 1& 1 \end{smallmatrix})$ & 
$(\begin{smallmatrix} 1 & 1 \\ 0 & 0 \end{smallmatrix})$ / $(\begin{smallmatrix} 0 & 1 \\ 0 & 1 \end{smallmatrix})$ \\[2.5pt] \hline
Gen.\ CNOT & $C(X,X)$ & $C(Y,Y)$ & $C(Z,Z)$ & $C(X,Z)$ & $C(Y,X)$ & $C(Z,Y)$ \\[2.5pt] \hline
\end{tabular}
\caption{Noninvertible tableau elements and the corresponding generalized CNOT gates produced by the universal construction.}
\label{table:noninvert}
\end{table}

\begin{figure}
\centering
\mbox{
\Qcircuit @C=2em @R=1.5em {
& \gate{\;b\;} & \qw  \\
& \gate{b^*} \qwx & \qw
}}
\caption{Circuit diagram for $\CNOT(b,1,2)$ gate.}
\label{fig:b_CNOT}
\end{figure}

Finally, we would like to have a direct way to compose two circuits by a simple operation on their tableaux.  Suppose we wish to compute the composition of circuits $C_1$ and $C_2$.  To compute the tableau of $C_2 \circ C_1$, we must compute, for each Pauli basis element $p_j$, the product $C_2 C_1 p_j C_1^\dag C_2^\dag$.  First consider the $j$th row of the tableau for $C_1$, which gives
$$
C_1 p_j C_1^{\dagger} = \alpha_j \prod_{k=1}^{2n} p_k^{M_{jk}^{(1)}},
$$
where $M^{(1)}$ is the binary representation of $\Tstar(C_1)$, and $\alpha_j$ is the phase.  Similarly,
$$
C_2 p_j C_2^{\dagger} = \beta_j \prod_{k=1}^{2n} p_k^{M_{jk}^{(2)}}.
$$
Therefore, 
\begin{align*}
C_2 C_1 p_j C_1^\dag C_2^\dag &= C_2\left( \alpha_j \prod_{k=1}^{2n} p_k^{M_{jk}^{(1)}} \right) C_2^\dag = \alpha_j \prod_{k=1}^{2n} \left(C_2 p_k^{M_{jk}^{(1)}} C_2^\dag \right)
=  \alpha_j \prod_{k=1}^{2n} \left(\beta_k \prod_{\ell=1}^{2n} p_\ell^{M_{k\ell}^{(2)}}\right)^{M_{jk}^{(1)}} \\
&= \alpha_j \prod_{k=1}^{2n} \left(\beta_k^{M_{jk}^{(1)}}\prod_{\ell=1}^{2n} p_\ell^{M_{jk}^{(1)} M_{k\ell}^{(2)} }\right) \propto \prod_{\ell=1}^{2n} \prod_{k=1}^{2n} p_\ell^{M_{jk}^{(1)} M_{k\ell}^{(2)} }
\propto \prod_{\ell=1}^{2n} p_\ell^{\bigoplus_{k=1}^{2n} M_{jk}^{(1)} M_{k\ell}^{(2)} } \\
&=  \prod_{\ell=1}^{2n} p_\ell^{[M^{(1)} M^{(2)}]_{j\ell} } 
\end{align*}

Notice that this implies $\Tstar(C_2 \circ C_1) = \Tstar(C_1) \Tstar(C_2)$.  Since it is cumbersome to write out explicitly, we did not include the exact phases in the above calculation.  Nevertheless, one can compute the phase bits by tracking the intermediate steps in the above calculation, which includes the multiplication of Pauli basis elements.

%!TEX root = ../full_paper.tex

\section{Classes}
\label{sec:classes}

The class generated by a set of Clifford gates is the collection of Clifford operations which can be constructed from circuits built from those gates. Formally, define a \emph{class} $\mathcal{C}$ to be a set of Clifford gates satisfying the following four rules:
\begin{enumerate}
\item \textbf{Composition Rule} $\mathcal{C}$ is closed under composition of gates.  If $f, g \in \CC$ are gates on the same number of qubits, then $f \circ g \in \CC$.
\item \textbf{Tensor Rule} $\mathcal{C}$ is closed under tensor product of gates.  If $f, g \in \CC$, then $f \otimes g \in \CC$.
\item \textbf{Swap Rule} $\mathcal{C}$ contains the $\operatorname{SWAP}$ gate.
\item \textbf{Ancilla Rule} $\mathcal{C}$ is closed under ancillas. If $f \in \CC$ and there exists $g$ such that 
$$f(\ket{x} \otimes \ket{\psi}) = g(\ket{x}) \otimes \ket{\psi}$$
 for some $\ket{\psi}$ and for all inputs $\ket{x}$ (up to a global phase), then $g \in \CC$. 
\end{enumerate}

Given that each class is supposed to capture those gates which can be built from a circuit, the operations of composition and tensor product are completely natural. Let us now spend some time to justify the swap and ancilla rules.

First, the swap rule allows us to consider gates without needing to specify the qubits on which they must be applied.  Indeed, we can relabel the input wires however we like. There are natural alternatives to this rule, e.g., allow reordering the qubits on each gate. The classification will be more complicated under these definitions, and we feel that adding the SWAP gate is cleanest way to address the fact that a reordering of qubits is inherently uninteresting.\footnote{From a category-theoretic perspective, the addition of the SWAP gate implies that our classes correspond to dagger \emph{symmetric} monoidal groupoids, which may feel more natural to some readers \cite{selinger2007dagger}.}

Second, the ancilla rule allows us to have useful workspace for our computation. Notice that the other rules are ``additive'' in the sense that the number of qubits in the circuit never decreases---we need a rule to reduce the number of qubits, otherwise there is no way that, e.g., a $3$-qubit gate could ever generate a $2$-qubit gate. This would lead to hierarchies of classes that differ only on gates with a small number of qubits (e.g., classes that are identical except on $1$-qubit and $2$-qubit gates), solely because of the size of the available generators. 

Instead, we permit quantum ancillas which can be used by the circuit, but must lead to a unitary transformation on the remaining non-ancilla (input) qubits.  To guarantee unitarity, notice that the ancilla qubits must not be entangled with the input qubits at the end of the circuit.  Furthermore, the ancillas cannot change in an input-dependent manner (e.g., consider the circuit that simply swaps in the input qubits with the ancilla qubits).  The ancilla rule attempts to capture these restrictions:  any ancilla qubits used during the computation must be returned to their initial configuration at the end of the computation. Notice that this rule coincides with the notion of a \emph{catalyst} (see, e.g.,  \cite{beverland2020lower, jonathan1999entanglement}).

One might assume that we could increase the power of ancillas by letting them change over the course of the computation, as long as the change is independent of the input (that is, from some constant initial state to some possibly different constant final state). Indeed, in the classification of reversible gates \cite{ags:2015}, these ``loose ancillas'' collapse a few pairs of classes.  Nevertheless, we show that even with loose ancillas, our classification (as presented) still holds.\footnote{When we formally define class invariants in Section~\ref{sec:invariants} (see Theorem~\ref{thm:trickysigns} in particular), we will see that the invariants hold under the loose quantum ancilla model, and therefore hold for all weaker models.}

% These ancilla inputs can be viewed as the workspace of the computation.  If ancillary qubits are not allowed, then a gate cannot be used to generate any smaller gate because there are not enough inputs to apply the gate. Furthermore, if we want to apply a Clifford operation as a subroutine of a general quantum computation, then we need that the ancilla states do not depend on the input.  Otherwise, they would destroy the quantum coherence of the computation.

Next, we have the question of how the ancillas are initialized. The weakest assumption one could make is that we have no control: the ancillas are initialized to an unknown state, and must be returned to this state at the end of the circuit. %This is somewhat artificial since we can, at the very least, initialize the workspace to the all-zeros state.  Other classifications suggest that without this assumption the problem becomes \emph{dramatically} more difficult.
This is unreasonably harsh, since it is almost always practical to initialize a qubits to the $\ket{0}$ state, and there is reason to believe the classification is \emph{dramatically} more difficult without some control over the state of the ancillas.\footnote{For example, the lattice of classical many-to-one functions over bits is finite when we allow 0/1 inputs to any function, but infinite when we do not allow such freedom  \cite{post, ags:2015}.} A slightly stronger assumption would be to allow ancillas initialized to computational basis states, but this would break symmetry by introducing a bias towards the $Z$-basis.\footnote{We can fix the bias by allowing all single-qubit ancillary states: $\ket{0},\ket{1},\ket{+},\ket{-1},\ket{i},\ket{-i}$.  This introduces new classes such as $\class{\theta_{X+Z} \otimes \theta_{X+Z}}$, but we leave the classification under these assumptions as an open problem.}  

A next natural step would be to permit ancillas initialized to arbitrary stabilizer states. Although this would appear to be circular (i.e., Clifford gates are necessary to implement stabilizer states, which we then use as ancillas in Clifford circuits), the reusability of the ancilla states implies that even if the states are difficult to construct, at least we only have to construct them once. Unfortunately, we are unable to complete the classification in its entirety under this ancilla model. However, we have reduced the problem to finding a single stabilizer state which is stabilized by $\Gamma$ and a permutation, and we conjecture that such a state exists (see Section~\ref{sec:open_problems}). Moreover, the conjectured classification matches the one we will present (under a stronger ancilla model).

Finally, we arrive at our chosen model, that is, ancillas initialized to arbitrary quantum states. A priori, these states could be arbitrarily large, and arbitrarily complicated to construct, which is clearly undesirable. It turns out, however, the classification only requires finitely many one or two qubit states, in particular, the eigenstates of the single-qubit Clifford gates and states that are locally equivalent to the Bell state. For comparison, we have determined the reversible gate lattice under quantum ancillas in Appendix~\ref{app:classical_quantum}, and observe that, in some cases, arbitrarily large, entangled states are necessary.

It is worth noting that there is a long line of work showing that weak gate sets, including Clifford gates, are universal for quantum computation when given access to magic states \cite{bravyi:2005, rbb, shor:1996}. Importantly, these magic states do \emph{not} need to be preserved after the computation. Conversely, under our model, we show that arbitrary quantum ancillas cannot boost the power of Clifford gates beyond the Clifford group.

Let's now see how the rules justify natural properties one would want from a class:
%\begin{proposition}
%Let $\mathcal{C}$ be a class of Clifford gates. Then $\mathcal{C}$ contains the $n$ qubit identity gate for any $n$. 
%\end{proposition}
%\begin{proof}
%First, $\mathcal{C}$ contains $\operatorname{SWAP}$. It follows that $\mathcal{C}$ contains the two qubit identity gate since it is the composition $\operatorname{SWAP} \circ \operatorname{SWAP}$. By the ancilla rule, we can remove a qubit from the two qubit identity using \emph{any} one-qubit state. Hence, the one qubit identity gate is in $\mathcal{C}$. Finally, $\mathcal{C}$ must contain the $n$-qubit identity gate because it is the tensor product of $n$ one-qubit gates. 
%\end{proof}
%
%\begin{proposition}
%Let $\mathcal{C}$ be a class of Clifford gates. For any $g \in \mathcal{C}$, the inverse, $g^{-1}$, belongs to $\mathcal{C}$. 
%\end{proposition}
%\begin{proof}
%Consider the sequence 
%$$
%g, g \circ g, g \circ g \circ g, \ldots, g^{n}, \ldots.
%$$
%Since there are finitely many Clifford gates on $n$ qubits (certainly finitely many tableaux, and one gate per tableau), the sequence must eventually repeat. That is, $g^{i} = g^{j}$ for some $1 \leq i < j$. Since every Clifford gate has an inverse, we conclude that $1 = g^{0} = g^{j-i}$, and hence $g^{-1} = g^{j-i-1}$. In other words, $g^{-1}$ is a Clifford gate, and $g^{-1}$ is a (positive) power of $g$ and therefore in $\mathcal{C}$. 
%\end{proof}

\begin{proposition}
Let $\mathcal{C}$ be a class of Clifford gates. Then $\mathcal{C}$ contains the $n$-qubit identity gate for all $n \geq 1$, and $\mathcal{C}$ is closed under inverses.
\end{proposition}
\begin{proof}
    All classes contain $\operatorname{SWAP}$, and by composition it contains $\operatorname{SWAP} \circ \operatorname{SWAP}$, which is the two-qubit identity. By the ancilla rule, we can remove a qubit to get the one-qubit identity, and then by the tensor rule we get the $n$-qubit identity.
    
    Now suppose $g \in \mathcal{C}$ is an $n$-qubit Clifford gate. Since the $n$-qubit Clifford group is finite, $g$ has finite order: $g^{r} = I$. If $r > 1$ then we can construct the $(r-1)$-fold composition, $g^{r-1}$, which is the inverse since $g \circ g^{r-1} = g^{r} = I$. Otherwise, $g = g^{-1} = I$.  
\end{proof}

The most practical way to talk about a class $\mathcal{C}$ is by its \emph{generators}. We say a set of gates $G$ \emph{generates} a class $\mathcal{C}$ if $G \subseteq \mathcal{C}$ and every class containing $G$ also contains $\mathcal{C}$. We introduce the notation $\class{ \cdot }$ for the class generated by a set of gates.  Similarly, we say that $G$ generates a specific gate $g$ if $g \in \class{G}$.

Our goal is therefore to identify all Clifford gate classes, determine their generators, and diagram the relationships between classes. As it turns out, there are $57$ different classes, which we have split across Figure~\ref{fig:degeneratelattice} (which contains the classes with single-qubit generators) and Figure~\ref{fig:nondegeneratelattice} (which contains the multi-qubit classes). Each class is labelled by a set of generators for that class, except for $\ALL$, the class of all Clifford gates; $\top$, the class of all single-qubit Clifford gates; and $\bot$, the class generated by the empty set.  Additionally, we abbreviate some class names in Figure~\ref{fig:degeneratelattice}: 
\begin{itemize}
\item $\theta_{+++}$, $\theta_{+--}$, $\theta_{-+-}$, $\theta_{--+}$ denote the single-qubit classes containing $\Gamma_{+++}$, $\Gamma_{+--}$, $\Gamma_{-+-}$, and $\Gamma_{--+}$ respectively, and three $\theta$ gates each, as indicated by the gray lines. 
\item $\theta_{xy}$ abbreviates $\theta_{x+y}$ or $\theta_{x-y}$ (it contains both) and similarly for $\theta_{xz}$ and $\theta_{yz}$.
\end{itemize}
% Some of the lines in Figure~\ref{fig:degeneratelattice} are gray and dotted, not for any technical reason, but because the lattice would be unreadable otherwise. 

In Figure~\ref{fig:nondegeneratelattice} each class includes the label of the single-qubit subgroup, even when not all of the single-qubit generators are necessary to generate the class. This is intended to make the relationship between the degenerate and non-degenerate lattices clearer. For example, $\tfour$ generates the Pauli group, $\P$, on its own (Lemma~\ref{lem:t4_generates_paulis}), but we label the class $\class{ \tfour, \P }$ to make it clear that the class $\class{\tfour}$ is above $\class{\P}$ in the lattice.

%!TEX root = ../full_paper.tex
\newcommand{\subinv}{\mathscr{I}}
\newcommand{\perinv}{\mathscr{P}}

\section{Invariants}
\label{sec:invariants}

Until now, we have defined each class in terms of the generators for that class. It turns out that each class can also be characterized as the set of all gates satisfying a collection of invariants. Section~\ref{sec:equivalence} formalizes this equivalence.  This section focuses on the form of the invariants themselves.

Informally, an \emph{invariant} is a property of gates, readily apparent from their tableaux, which is preserved by the circuit building operations. In other words, if a collection of gates all satisfy a particular invariant then any circuit constructed from those gates must also satisfy the invariant. All our invariants are formally defined from the tableaux, but for now we give the following informal descriptions to make the intuition for each invariant clear.

\begin{description}
\item[$X$-, $Y$-, or $Z$-preserving:] We say a Clifford gate is \emph{$Z$-preserving} if it maps $Z$-basis states to $Z$-basis states, possibly with a change of phase. The $Z$-preserving gates include all (classical) reversible gates (e.g., $X$, $\operatorname{CNOT}$, and $\tfour$), gates which only manipulate the sign (e.g., $\R_Z$ and $\CSIGN$), and combinations of the two.

Symmetrically, there are $X$-preserving gates mapping $X$-basis states to $X$-basis states, and $Y$-preserving gates gates mapping $Y$-basis states to $Y$-basis states.  Our definitions of classes, gates, invariants, etc., are completely symmetric with respect to $X$, $Y$ and $Z$ basis, so if some gate or class is $X$-preserving (resp. $Y$-preserving or $Z$-preserving), then there must be a corresponding gate or class which is $Y$-preserving (resp. $Z$-preserving or $X$-preserving). We will often appeal to this symmetry to simplify proofs.

Note that a gate can be any combination of $X$-, $Y$-, and $Z$-preserving. For instance, $\tfour$ is $X$-, $Y$-, and $Z$-preserving; $\operatorname{CNOT}$ is $X$-preserving and $Z$-preserving but not $Y$-preserving (similarly $C(X,Y)$ and $C(Y,Z)$ fail to be $Z$-preserving and $X$-preserving, respectively); $\R_Z$ is $Z$-preserving only (similarly $\R_X$ is $X$-preserving and $\R_Y$ is $Y$-preserving); and $\Gamma$ is not $X$-, $Y$-, or $Z$-preserving. 
\item[Egalitarian] We say an $n$-qubit gate $U$ is \emph{egalitarian} if 
$$
\Tstar(\Gamma^{\otimes n} U (\Gamma^\dagger)^{\otimes n}) = \Tstar(U).
$$
Egalitarian gates have no preferred basis, since conjugation by $\Gamma$ cycles $X$ to $Y$, $Y$ to $Z$, and $Z$ to $X$. More concretely, if an egalitarian gate $U$ maps Pauli string $P$ to $Q = U P U^{\dag}$ under conjugation, then $U$ maps $\Gamma^{\otimes n} P (\Gamma^{\dag})^{\otimes n}$ to $\pm \Gamma^{\otimes n} Q (\Gamma^{\dag})^{\otimes n}$. 
%That is, the gate is egalitarian if it is fixed by an $X/Y/Z$ symmetry, up to a Pauli operation.  This is the symmetry arising from conjugating all qubits by $\Gamma$ (which cycles $X$ to $Y$, $Y$ to $Z$, and $Z$ to $X$). 
%If egalitarian operation $U$ maps Pauli string $P$ to $Q = U P U^\dag$ under conjugation, then $U$ maps $\Gamma P \Gamma^\dag$ to
%$$U \Gamma P \Gamma^\dag U^\dag \propto \Gamma U \Gamma^\dag \Gamma P \Gamma^\dag \Gamma U^\dag \Gamma^\dag =  \Gamma U  P U^\dag \Gamma^\dag =  \Gamma Q \Gamma^\dag.$$
The Pauli matrices, $\Gamma$, and $\tfour$ are examples of egalitarian gates. 
\item[Degenerate:] We say a gate is \emph{degenerate} if each input affects only one output. More precisely, when applying the gate to a string of Paulis, changing one Pauli in the input will change exactly one Pauli in the output. All single-qubit gates are degenerate, and all degenerate gates can be composed of single-qubit gates and SWAP gates.
\item[$X$-, $Y$-, or $Z$-degenerate:] A gate is \emph{$Z$-degenerate} if it is $Z$-preserving and flipping any bit of a classical ($Z$-basis) input to the gate causes exactly one bit of the output to flip. The gate may or may not affect the phase. This class includes several $Z$-preserving single-qubit gates, like $\R_Z$, the Pauli operations, and $\theta_{X+Y}$. It also includes $\CSIGN$ because this gate \emph{only} affects phase, but $\operatorname{CNOT}$ is not $Z$-degenerate because flipping the control bit changes both outputs. Notice that $\CSIGN$ is $Z$-degenerate, but not degenerate. We define $X$-degenerate and $Y$-degenerate symmetrically. 
\item[$X$-, $Y$-, or $Z$-orthogonal:] A gate $G$ is \emph{$Z$-orthogonal} if it can be built from $\tfour$ and $Z$-preserving single-qubit gates. The term ``orthogonal'' comes from the fact that $\tfour$ is an orthogonal linear transformation in the $Z$-basis, but not all $Z$-orthogonal gates are literally orthogonal transformations in the $Z$-basis (see, for example, Lemma~\ref{lem:t4_generates_paulis}). Similarly for $X$-orthogonal and $Y$-orthogonal.
\item[Single-Qubit Gates:]
There are thirty different classes of single-qubit gates. All of these classes are degenerate, and some can be distinguished by the other invariants above. However, many single-qubit invariants depend on the phase bits of the tableau. For instance, the tableau of $\theta_{X+Y}$, $\theta_{X-Y}$, and $\R_Z$ all have the same matrix part, $(\begin{smallmatrix} 1 & 1 \\ 0 & 1 \end{smallmatrix})$, but generate three distinct classes. One can write down explicit invariants for these classes where the phase bits are correlated to the tableau entries, but in most cases we present a single-qubit class as a subgroup of the symmetries of the cube/octahedron, as shown in Figure~\ref{fig:cube}. 
\end{description}

\subsection{Formal invariants}

An \emph{invariant} is a property of tableaux which is preserved by the four circuit-building rules. 
\begin{description}
\item[Swap Rule:] Every class contains the $\operatorname{SWAP}$ gate, so every invariant we propose must be satisfied by the tableau for $\operatorname{SWAP}$. 
\item[Composition Rule:] If the invariant holds for two gates, then it must hold for their composition. We have seen that the tableau for the composition of two gates is essentially the matrix product of the two tableau, except for the phase bits (which are significantly more complicated to update).
\item[Tensor Rule:] The tensor product of two gates satisfying the invariant must also satisfy the invariant. Note that the tableau of the tensor product is the direct sum of the tableaux, and phase bits are inherited from the sub-tableaux in the natural way.
\item[Ancilla Rule:] The invariant must be preserved when some qubits are used as ancillas. It turns out the ancilla operation reduces the tableau to a submatrix (of non-ancilla rows and columns) and under certain conditions, the corresponding subset of the phase bits. This is somewhat technical, so we prove it in Theorem~\ref{thm:trickysigns} below.
\end{description}

\begin{theorem}
\label{thm:trickysigns}
Let $G$ be a Clifford gate on $n$ qubits, and suppose there exist states $\ket{\psi}$ and $\ket{\psi'}$ such that 
$$
G(\ket{x} \otimes \ket{\psi}) = H(\ket{x}) \otimes \ket{\psi'}
$$
for all $\ket{x}$, for some unitary $H$ on $m$-qubits. In particular, this is true if we use the ancilla rule to reduce $G$ to $H$, where $\ket{\psi} = \ket{\psi'}$ is the ancilla state. Then 
\begin{enumerate}
\item $H$ is a Clifford operation,
\item $\Tstar(H)$ is obtained by removing the rows and columns corresponding to the ancilla bits from $\Tstar(G)$,
\item If a row of $\Tstar(H)$ is obtained by removing only zeros, then the phase bit for that row is the same in $G$ and $H$.
% If every bit (in $\Tstar(G)$ as a binary matrix) we remove from a row is zero, then the phase bit for that row is the same in $G$ and $H$.
\end{enumerate}
\end{theorem}
\begin{proof}
Let $P \in \P_m$. Then for any $\ket{x}$, 
$$
G(P\ket{x} \otimes \ket{\psi}) = H(P \ket{x}) \otimes \ket{\psi'}.
$$
On the other hand, $G$ is a Clifford gate, so conjugating the Pauli string $P \otimes I^{n-m}$ by $G$ produces $\alpha Q \otimes R$ for Pauli strings $Q \in \P_m$ and $R \in \P_{n-m}$ and phase $\alpha \in {\pm 1}$. Equivalently, 
$$
G (P \otimes I^{n-m}) = \alpha (Q \otimes R) G.
$$
It follows that 
\begin{align*}
H(P \ket{x}) \otimes \ket{\psi'} &= G (P \ket{x} \otimes \ket{\psi}) \\
&= \alpha (Q \otimes R) G(\ket{x} \otimes \ket{\psi}) \\
&= \alpha (Q \otimes R) \left( H(\ket{x}) \otimes \ket{\psi'} \right) \\
&= \alpha_1 Q H(\ket{x}) \otimes \alpha_2 R \ket{\psi'},
\end{align*}
for some choice of $\alpha_1, \alpha_2 \in \mathbb C$ for which $\alpha_1 \alpha_2 = \alpha$, $H(P \ket{x}) = \alpha_1 Q H(\ket{x})$, and $\ket{\psi'} = \alpha_2 R \ket{\psi'}$. We see that $\ket{\psi'}$ is an eigenvector with eigenvalue $\alpha_2^{-1}$, so our choice of $\alpha_2$ (and hence $\alpha_1$) is unique. Since $R$ is a Pauli with eigenvalues $\pm 1$, it follows that $\alpha_2 = \pm 1$. %This last equation, $\ket{\psi'} = \alpha_2 R \ket{\psi'}$ shows that the choice of $\alpha_2$ is fixed: if $\alpha_2' R$ also stabilized $\ket{\psi'}$ then we would have 
%$$
%\alpha_2 R \ket{\psi'} = \ket{\psi'} = \alpha_2' R \ket{\psi'},
%$$
%so $\alpha_2 = \alpha_2'$. Furthermore, $\ket{\psi'} = (\alpha_2 R)(\alpha_2 R) \ket{\psi'} = \alpha_2^2 \ket{\psi'}$, so $\alpha_2 \in \{\pm 1\}$.

Therefore, we have $\alpha_1 \in \{\pm 1\}$ and $H P = \alpha_1 Q H$, which implies $H P H^{\dagger} = \alpha_1 Q$.  That is, since $P$ was arbitrary, the conjugation of a Pauli element by $H$ is always another Pauli element, so $H$ is a Clifford gate. In the special case that $P$ (and therefore $P \otimes I^{n-m}$) is a Pauli basis element, then $\alpha Q \otimes R$ is represented in the row of the binary tableau of $G$. We keep the bits representing $Q$ in the tableau for $H$, which is why $\Tstar(H)$ is a submatrix of $\Tstar(G)$.  The phase in the new tableau is $\alpha_1$. In the special case $R = I^{n-m}$, the only eigenvalue is $1$ so $\alpha_2 = 1$ and hence $\alpha_1 = \alpha$. That is, the phase for the corresponding row of the tableau for $H$ is inherited from $G$. 
\end{proof}

As a direct consequence of these rules, our invariants take on a distinctly algebraic flavor. Let us consider, for the sake of illustration, invariants that depend only on the matrix part of the tableau and ignore the phase bits. Then an invariant is equivalent to a set of matrices closed under the four rules above. In particular, the matrices do form a group under multiplication as a consequence of the composition rule (and the fact that every gate has finite order). 

On the other hand, not every group of matrices will correspond to an invariant. For instance, due to the swap rule, the group of matrices must also be closed under arbitrary reordering of the rows and columns. This eliminates, e.g., the group of upper triangular matrices. Similarly, the ancilla rule excludes the special orthogonal group. In the end, we are left with just two kinds of matrix groups which lead to invariants: 
\begin{description}
\item[Subring Invariants] Matrices with elements restricted to a particular subring of $\ring$ (analogous to the real matrices, integer matrices, etc.)
\item[Permutation Invariants] Permutation matrices, except where each $1$ entry can be any one of a subset of invertible elements, and each $0$ entry comes from a collection of non-invertible elements. 
\end{description}

Now we are ready to present formal definitions for these invariants, and show that they really are preserved by the circuit-building rules. 

\subsection{Subring invariants}

The first kind of invariant restricts the entries of the tableau to a subring of $\ring$. That is, given a subring $\subring \subseteq \ring$, a gate satisfies the invariant $\subinv(\subring)$ if and only if all entries of the tableau are in $\subring$.\footnote{There are several interesting works that connect the elements of the matrix representation of a unitary, to the set of gates that generate it.  For example, every multi-qubit unitary with elements in $\mathbb Z[1/\sqrt 2, i]$ corresponds to a circuit built from Clifford + $T$ gates \cite{kliuchnikov2012fast, giles2013exact}. There is a classification of gates corresponding to subrings of $\mathbb Z[1/\sqrt 2, i]$ \cite{amyglaudellross:2020}.  We stress, however, that our classification depends on the the ring elements of the \emph{tableau}.} There are twelve classes, all near the top of the lattice, of the form 
$$
\mathcal{C} = \{ \text{All gates satisfying $\subinv(\subring)$} \},
$$
corresponding to all 12 subrings of $\ring$ listed below. 

\begin{itemize}
\item The entire ring, $\ring$, is technically a subring of itself, and $\subinv(\ring)$ is the trivial invariant satisfied by all Clifford gates. Notice that not \emph{every} matrix over $\ring$ gives a valid tableau because it must still be symplectic. 
\item There are four maximal proper subrings of $\ring$:
\begin{align*}
\ring_{X} &= \{ (\begin{smallmatrix} a & 0 \\ c & d \end{smallmatrix}) : \{ a,c,d \} \in \{ 0, 1 \} \}, \\
\ring_{Y} &= \{ (\begin{smallmatrix} a & b \\ c & d \end{smallmatrix}) : \{ a,b,c,d \} \in \{ 0, 1 \}, a+b+c+d=0 \}, \\
\ring_{Z} &= \{ (\begin{smallmatrix} a & b \\ 0 & d \end{smallmatrix}) : \{ a,b,d \} \in \{ 0, 1 \} \}, \\
\ring_{E} &= \{ (\begin{smallmatrix} a & b \\ b & a+b \end{smallmatrix}) : \{ a,b \} \in \{ 0, 1 \} \}.
\end{align*}
Our formal definition for $Z$-preserving gates is the invariant $\subinv(\ring_Z)$. The fact that the lower left entry is $0$ implies that the gate maps Pauli strings of $I$ and $Z$  to strings of $I$ and $Z$. Hence, $Z$-basis strings are mapped to $Z$-basis strings. Similarly, the $X$-preserving and $Y$-preserving invariants are $\subinv(\ring_X)$ and $\subinv(\ring_Y)$ respectively. The egalitarian invariant, $\subinv(\ring_E)$, comes from the subring $\ring_{E}$. 
\item The intersection of two subrings is itself a subring, giving us exactly four more subrings ($\ring_X \cap \ring_Y$, $\ring_X \cap \ring_Z$, $\ring_Y \cap \ring_Z$, and $\ring_X \cap \ring_Y \cap \ring_Z$) since the intersection of $\ring_E$ with any of the others is 
$$
\ring_X \cap \ring_Y \cap \ring_Z = \{ (\begin{smallmatrix} 0 & 0 \\ 0 & 0 \end{smallmatrix}), (\begin{smallmatrix} 1 & 0 \\ 0 & 1 \end{smallmatrix}) \},
$$
the trivial ring. 
\item Three more subrings are obtained by taking only self-conjugate elements of $\ring_X$, $\ring_Y$, and $\ring_Z$ respectively. An element $(\begin{smallmatrix} a & b \\ c & d \end{smallmatrix})$ is \emph{self-conjugate} if 
$$
(\begin{smallmatrix} a & b \\ c & d \end{smallmatrix}) = (\begin{smallmatrix} a & b \\ c & d \end{smallmatrix})^*,
$$
or equivalently, $a = d$. These invariants correspond to the $X$-orthogonal (i.e., $\class{\tfour, \P, \R_X}$), $Y$-orthogonal (i.e., $\class{\tfour, \P, \R_Y}$), and $Z$-orthogonal (i.e., $\class{\tfour, \P, \R_Z}$) classes respectively. 
\end{itemize}

\begin{theorem}
For any subring $\subring \subseteq \ring$, the property $\subinv(\subring)$ is an invariant. That is, the set of matrices over $\subring$ respect the circuit building operations. 
\end{theorem}
\begin{proof}
Every subring contains $(\begin{smallmatrix} 0 & 0 \\ 0 & 0 \end{smallmatrix})$ and $(\begin{smallmatrix} 1 & 0 \\ 0 & 1 \end{smallmatrix})$ by definition, and therefore the tableau (phase bits omitted) of the $\operatorname{SWAP}$ gate, 
$$
\left(
\begin{array}{cc|cc} 
0 & 0 & 1 & 0 \\
0 & 0 & 0 & 1 \\
\hline
1 & 0 & 0 & 0 \\
0 & 1 & 0 & 0 \\
\end{array}
\right)
$$
satisfies $\subinv(\subring)$. 

Matrix multiplication is a polynomial in the entries of the two matrices, so composition cannot produce entries outside the subring. Similarly, combining tableaux with tensor products or reducing tableaux to submatrices via ancillas does not introduce any new ring elements (see Theorem~\ref{thm:trickysigns}); those operations only use elements already present in the tableau. We conclude that $\subinv(\subring)$ is an invariant for any subring $\subring$. 
\end{proof}

\subsection{Permutation invariants}

The \emph{permutation invariants} get their name from the matrix part of their tableaux, which is required to have the structure of a permutation matrix. That is, every row (or column) has exactly one element which is invertible, and the others are non-invertible. Permutation invariants are also sensitive to phase bits. It is natural to associate the unique invertible element in a row with the phase bits for that row, giving the tableau of a single-qubit gate. A permutation invariant $\perinv(G, S)$ is defined by the set of single-qubit gates $G$ which can be obtained in this way, and the set of non-invertible elements $S$ used to fill the rest of the tableau. In other words, a tableau satisfies $\perinv(G, S)$ if all entries are from $S$ except exactly one entry per row which, when combined with the phase bits for the row, is the tableau of some gate in $G$. 

Note that not all pairs of sets $(G, S)$ produce an invariant. For instance, circuit-building operations will fail to preserve $\perinv(G, S)$ if $G$ is not a group. The exact relationship between $G$ and $S$ required to produce an invariant is difficult to write down. Roughly speaking, products of elements in $S$ should be zero, products of elements in $G$ should remain in $G$, and products between $S$ and $\Tstar(G)$ should be manageable in some sense. Theorem~\ref{thm:permutationinvariant} gives a list of $\perinv(G, S)$ invariants, which will turn out to be exhaustive by Theorem~\ref{thm:final}, the culminating theorem of this paper.

\begin{theorem}
\label{thm:permutationinvariant}
We prove that the following permutation invariants are indeed invariant under the circuit-building operations. Let $G$ be a group of single-qubit gates.
\begin{enumerate}
\item Then 
$$\perinv(G, \{ (\begin{smallmatrix} 0 & 0 \\ 0 & 0 \end{smallmatrix}) \})$$
is an invariant for $\class{G}$. All thirty degenerate classes are characterized by invariants of this form.
\item If $\class{X} \subseteq \class{G} \subseteq \class{\P, \R_X}$ then 
$$
\perinv(G, \{ (\begin{smallmatrix} 0 & 0 \\ 0 & 0 \end{smallmatrix}), (\begin{smallmatrix} 0 & 0 \\ 1 & 0 \end{smallmatrix}) \})
$$
is an invariant for $\class{ C(X,X), G}$. These invariants characterize the five $X$-degenerate classes.
\item If $\class{Y} \subseteq \class{G} \subseteq \class{\P, \R_Y}$ then 
$$
\perinv(G, \{ (\begin{smallmatrix} 0 & 0 \\ 0 & 0 \end{smallmatrix}), (\begin{smallmatrix} 1 & 1 \\ 1 & 1 \end{smallmatrix}) \})
$$
is an invariant for $\class{C(Y,Y), G}$. These invariants characterize the five $Y$-degenerate classes.
\item If $\class{Z} \subseteq \class{G} \subseteq \class{\P, \R_Z}$ then 
$$
\perinv(G, \{ (\begin{smallmatrix} 0 & 0 \\ 0 & 0 \end{smallmatrix}), (\begin{smallmatrix} 0 & 1 \\ 0 & 0 \end{smallmatrix}) \})
$$
is an invariant for $\class{C(Z,Z), G}$. These invariants characterize the five $Z$-degenerate classes.
\end{enumerate}
\end{theorem}
\begin{proof}
Let $\perinv(G,S)$ be one of the invariants above. Let $M = \{ \Tstar(g) : g \in G \}$ be the set of matrices from tableaux in $G$. 
In all cases, $S$ contains $(\begin{smallmatrix} 0 & 0\\ 0 & 0 \end{smallmatrix})$, and $G$ contains the single-qubit identity operation,
$$
\left(
\begin{array}{cc|c} 
1 & 0 & 0 \\ 0 & 1 & 0
\end{array} \right),
$$
so $\operatorname{SWAP}$ satisfies the invariant. And clearly the direct sum of two tableaux in $\perinv(G,S)$ is still in $\perinv(G,S)$ for any $G$ and $S$.

Now consider the composition of two gates. Each entry in the tableau is a dot product of some row from one tableau with some column from the other. Hence, the entry is a sum of $S \times S$, $S \times M$, $M \times S$, or $M \times M$ products. Observe that the $S \times S$ products are all zero (for the particular sets $S$ above), so we may ignore those products. Recall that the row and column each contain exactly one entry in $M$, so depending on whether those entries align, we get either $SM + MS \subseteq S$ or $M^2 = M$. Furthermore, for any row in one tableau there is exactly one column in the other such that the invertible entries line up.  Therefore, exactly one entry in any row (or column) of the composition is in $M$ and the rest are in $S$. Clearly the matrix part of the tableau has the correct form for the invariant. 

We must also consider phase bits under composition. Recall that the phase bits associate with the invertible entries of the matrix to produce single-qubit gates. When we multiply two tableau, these single-qubit gates multiply to produce elements in $G$ (since $G$ is a group), as you would expect. If the non-invertible elements are all zero, then this is the only factor in determining phase bits, so the invariant is preserved by composition.

Now consider the phase bits in the case where $S$ contains nonzero elements, for instance, 
$$S = \{ (\begin{smallmatrix} 0 & 0 \\ 0 & 0 \end{smallmatrix}), (\begin{smallmatrix} 0 & 1 \\ 0 & 0 \end{smallmatrix}) \}.$$
Notice that in this case, both matrices in $S$ have zeros in the bottom row, and the invertible matrices are of the form $(\begin{smallmatrix} 1 & a \\ 0 & 1 \end{smallmatrix})$. Hence, every even-indexed row of the tableau (as a binary matrix) is all zeros except for one entry. Using the method of tableau composition in Section~\ref{sec:tableaux}, one can easily show that for these even-indexed rows, the phase bits are exactly what one would get by composing the invertible elements as gates in $G$. For the other half of the rows, the non-invertible elements may flip the phase bits. But we assume $G$ contains the Pauli element (in this case $Z$) which flips that sign, so the invertible elements and associated phase bits are still in $G$, therefore the invariant is preserved. The $X$- and $Y$-degenerate cases are similar. 

Last, we show that $\perinv(G,S)$ is preserved under ancilla operations. Recall that when we use ancillas, we remove the rows and columns corresponding to those bits. Clearly the elements of the submatrix are still in $M$ and $S$. There is a risk that the invertible element for some row could be in one of the removed columns, but if the submatrix is missing an invertible element in some row then the submatrix violates the symplectic condition and the ancilla rule must have been misapplied. Hence, only elements in $S$ are removed in the non-ancilla part of the tableau, and each row still contains exactly one entry in $M$. 

We appeal to Theorem~\ref{thm:trickysigns} for the phase bits. The theorem says that removing the ancillas can only change the sign for a row if there is a nonzero entry in the non-ancilla bits of the row that are removed. For example, if $S = \{ (\begin{smallmatrix} 0 & 0 \\ 0 & 0 \end{smallmatrix}), (\begin{smallmatrix} 0 & 1 \\ 0 & 0 \end{smallmatrix}) \}$ then only the top phase bit can change. But changing the top phase bit is the same as applying a $Z$, and for this case $Z$ is assumed to be in $G$, so the combination of the element in $M$ and the phase bits is still in $G$. Therefore the $Z$-degenerate $\perinv(G,S)$ are invariants, and the $X$-degenerate and $Y$-degenerate invariants follow by symmetry.

\end{proof}

%!TEX root = ../full_paper.tex

\section{Equivalence of Generator and Invariant Definitions}
\label{sec:equivalence}

We have now defined each class by a set of generators, and by an invariant, but have not yet shown that these definitions coincide. Below are a collection of lemmas which prove this for all classes in our lattice. Note that one direction is always trivial: it is easy to check that the generators defining a class satisfy a particular invariant, and therefore everything they generate (i.e., the class) must satisfy the invariant.  We encourage the reader to check these invariants against, say, the tableaux in Table~\ref{table:2_qubit_gates}.

For the other inclusion (i.e., every gate satisfying the invariant can be generated by the given generators), we start with an arbitrary gate $g$ satisfying the invariant, and apply gates in the class to $g$ to simplify its tableau step-by-step until it is the tableau of the identity operation. It follows that $A g B = I$ for circuits $A$ and $B$ in the class, which proves $g = A^{-1} B^{-1}$ is in the class. In many cases, the circuit derived this way is a \emph{canonical form} for the gate, and can be used to count the number of gates on $n$ qubits in a class.

Let us start with the degenerate classes. 
\begin{lemma}
\label{lem:degenerate}
Let $G$ be a group of single-qubit gates, and let $g$ be a gate satisfying the permutation invariant $\perinv(G, \{ (\begin{smallmatrix} 0 & 0 \\ 0 & 0 \end{smallmatrix}) \})$. Then there is a circuit for $g$ consisting of a permutation of the inputs followed by layer of single-qubit gates in $G$.
\end{lemma}
\begin{proof}
Consider the tableau for $g$. Each row or column has exactly one invertible element, so we can read off a permutation $\pi$ from the positions of those elements. Apply $\operatorname{SWAP}$ gates to $g$ to remove this permutation, and put the invertible elements on the diagonal. When we pair a diagonal element with the phase bits for that row, we get a single-qubit gate $g_i$ in $G$. Applying the inverse of this gate to qubit $i$ will zero the phase bits for that row, and make $(\begin{smallmatrix} 1 & 0 \\ 0 & 1 \end{smallmatrix})$ the diagonal entry. Once we do this for each row, we have the identity tableau, therefore $g$ is in $\class{G}$ and has a circuit of the desired form. 
\end{proof}

Next, we consider the $Z$-degenerate classes and, by symmetry, the $X$- and $Y$-degenerate classes.
\begin{lemma}
\label{lem:semidegenerate}
Let $G$ be a group of $Z$-preserving single-qubit gates such that $\class{Z} \subseteq \class{G} \subseteq \class{\P, \R_Z}$. Let $g$ be any gate satisfying the permutation invariant $\perinv(G, \{ (\begin{smallmatrix} 0 & 0 \\ 0 & 0 \end{smallmatrix}), (\begin{smallmatrix} 0 & 1 \\ 0 & 0 \end{smallmatrix}) \})$. Then there is a circuit for $g$ consisting of a layer of single-qubit gates (from $G$), a layer of $C(Z,Z)$ gates, and a permutation. 
\end{lemma}
\begin{proof}
Consider the tableau of $g$. We can read off a permutation $\pi$, and a single-qubit gate for each input. Assume we have removed those gates (i.e., we now consider the tableau of $g \pi^{-1} g_1^{-1} \cdots g_n^{-1}$), so the tableau has $(\begin{smallmatrix} 1 & 0 \\ 0 & 1 \end{smallmatrix})$ on the diagonal, all other entries are either $(\begin{smallmatrix} 0 & 1 \\ 0 & 0 \end{smallmatrix})$ or zero, and the phase bits are zero. 

The non-zero, off-diagonal entries in the matrix indicate the positions of $C(Z,Z)$ gates. Specifically, if the entry in row $i$ and column $j$ is nonzero then there is a $C(Z,Z)$ on qubits $i$ and $j$. Note that because the matrix part of the tableau is symplectic, the symmetric entry in row $j$ and column $i$ must also be non-zero. The remainder of the circuit consists of the set of $C(Z,Z)$ gates indicated by the non-zero, off-diagonal entries. Notice that $C(Z,Z)$ gates always commute, so their ordering does not matter.
\end{proof}

Now let us consider four $Z$-preserving classes which, when we consider symmetry (i.e., the $X$-preserving and $Y$-preserving equivalents) cover all but two of the remaining classes.
\begin{lemma}
\label{lem:eliminate}
Each of the classes $\class{\tfour, \P}$, $\class{\tfour, \P, \R_Z}$, $\class{C(Z,X), \P}$, and $\class{C(Z,X), \P, \R_Z}$ is the set of all gates corresponding to a subring invariant, where the subrings are 
\begin{align*}
\subring_1 &= \{ (\begin{smallmatrix} 0 & 0 \\ 0 & 0 \end{smallmatrix}), (\begin{smallmatrix} 1 & 0 \\ 0 & 1 \end{smallmatrix}) \}, \\
\subring_2 &= \{ (\begin{smallmatrix} 0 & 0 \\ 0 & 0 \end{smallmatrix}), (\begin{smallmatrix} 0 & 1 \\ 0 & 0 \end{smallmatrix}), (\begin{smallmatrix} 1 & 0 \\ 0 & 1 \end{smallmatrix}), (\begin{smallmatrix} 1 & 1 \\ 0 & 1 \end{smallmatrix}) \}, \\
\subring_3 &= \{ (\begin{smallmatrix} 0 & 0 \\ 0 & 0 \end{smallmatrix}), (\begin{smallmatrix} 1 & 0 \\ 0 & 0 \end{smallmatrix}), (\begin{smallmatrix} 0 & 0 \\ 0 & 1 \end{smallmatrix}), (\begin{smallmatrix} 1 & 0 \\ 0 & 1 \end{smallmatrix}) \}, \\
\subring_4 &= \{ (\begin{smallmatrix} 0 & 0 \\ 0 & 0 \end{smallmatrix}), (\begin{smallmatrix} 0 & 0 \\ 0 & 1 \end{smallmatrix}), (\begin{smallmatrix} 0 & 1 \\ 0 & 0 \end{smallmatrix}), (\begin{smallmatrix} 0 & 1 \\ 0 & 1 \end{smallmatrix}), (\begin{smallmatrix} 1 & 0 \\ 0 & 0 \end{smallmatrix}), (\begin{smallmatrix} 1 & 0 \\ 0 & 1 \end{smallmatrix}), (\begin{smallmatrix} 1 & 1 \\ 0 & 0 \end{smallmatrix}), (\begin{smallmatrix} 1 & 1 \\ 0 & 1 \end{smallmatrix}) \}.
\end{align*}
\end{lemma}
\begin{proof}
In all four classes, elements of the tableau are of the form $(\begin{smallmatrix} a & b \\ 0 & d \end{smallmatrix})$. Suppose $(\begin{smallmatrix} ? & ? \\ 0 & d_1 \end{smallmatrix})$ is the $i$th entry of some row, and $(\begin{smallmatrix} ? & ? \\ 0 & d_2 \end{smallmatrix})$ is the $j$th entry in the same row, where entries labeled by ``?'' are unconstrained. If we apply a $\operatorname{CNOT}$ gate from qubit $i$ to $j$, these entries will be of the form $(\begin{smallmatrix} ? & ? \\ 0 & d_1 \end{smallmatrix})$ and $(\begin{smallmatrix} ? & ? \\ 0 & d_1 + d_2 \end{smallmatrix})$ respectively. That is, the bottom right bits change as though we applied the $\operatorname{CNOT}$ gate to those bits. Since a $\tfour$ gate can be built from $\operatorname{CNOT}$ gates, it will (similarly) affect the bottom right bits as though we are applying a $\tfour$. 

Our strategy is to use either $\operatorname{CNOT}$ or $\tfour$ gates (depending on the class) to perform Gaussian elimination on the bottom right entries of the matrix elements. If we have access to $\operatorname{CNOT}$ gates then we literally apply Gaussian elimination, using $\operatorname{CNOT}$ to add one column to another, and using $\operatorname{SWAP}$ to exchange columns. 

If we only have $\tfour$ gates then we are in subring $\subring_1 \subseteq \subring_2$ or $\subring_2$, so $(\begin{smallmatrix} 1 & 0 \\ 0 & 1 \end{smallmatrix})$ and $(\begin{smallmatrix} 1 & 1 \\ 0 & 1 \end{smallmatrix})$ are the only elements with a $1$ in the bottom right position, and also the only invertible elements. It follows that the number of bottom right bits set to $1$ in a row is the same as the number of invertible elements, which must be odd because the matrix is symplectic. To reduce the number, we apply a $\tfour$ to three $1$ bits and a $0$ bit (note: we may add a zero bit by adding an ancilla, if necessary), which changes the $0$ to a $1$ and the $1$'s to $0$'s, reducing the number of $1$'s (or invertible elements) in the row by two. When there is a single $1$ left in the row, symplectic conditions imply that it is also the only $1$ left in that column, so we may ignore that row and column for the moment and continue to eliminate the rest of the matrix. 

Now suppose we have row reduced the matrix, using either $\operatorname{CNOT}$ or $\tfour$, so that the bottom right entry of every element is $0$, except along the diagonal where that bit $1$. At this point, the diagonal element is the only element in a row that can possibly be invertible, therefore the diagonal elements are of the form $(\begin{smallmatrix} 1 & b \\ 0 & 1 \end{smallmatrix})$. Similarly, the symplectic conditions imply that the off-diagonal elements are of the form $(\begin{smallmatrix} 0 & b \\ 0 & 0 \end{smallmatrix})$. In other words, the remaining tableau is $Z$-degenerate, since there is only one invertible element per row or column, and the off-diagonal elements are in $I = \{ (\begin{smallmatrix} 0 & 0 \\ 0 & 0 \end{smallmatrix}), (\begin{smallmatrix} 0 & 1 \\ 0 & 0 \end{smallmatrix}) \}$. We can use either Lemma~\ref{lem:degenerate} or Lemma~\ref{lem:semidegenerate} to find a circuit from the remainder, which is in either $\class{\P}$ or $\class{C(Z,Z),\P,\R_Z}$, depending on the class.
\end{proof}

There are only two classes remaining, $\ALL$ and $\class{\tfour,\P,\Gamma}$, which we handle specially. For the first, we appeal to Aaronson and Gottesman \cite{ag} who give an explicit decomposition for any Clifford gate into layers of $\operatorname{CNOT}$, Hadamard ($\theta_{X+Z}$ in our notation), and phase ($\R_{Z}$) gates.
\begin{lemma}
Any egalitarian gate $g$ can be generated in $\class{\tfour, \P, \Gamma}$. 
\end{lemma}
\begin{proof}
Egalitarian gates satisfy the invariant that all elements are in the subring 
$$
\{ (\begin{smallmatrix} 0 & 0 \\ 0 & 0 \end{smallmatrix}), (\begin{smallmatrix} 1 & 0 \\ 0 & 1 \end{smallmatrix}), (\begin{smallmatrix} 1 & 1 \\ 1 & 0 \end{smallmatrix}), (\begin{smallmatrix} 0 & 1 \\ 1 & 1 \end{smallmatrix}). \}
$$
In fact, this subring is isomorphic to $\mathbb{F}_4$, so it is a field. In particular, we will use that the only noninvertible element of the subring is zero.  The symplectic conditions on the tableau translate to it being unitary as a matrix over $\mathbb{F}_4$. 

Like the other $\tfour$ classes, we use Gaussian elimination on the tableau of $g$. Consider a row of the tableau. If the entry in some column is not the identity, then apply $\Gamma$ or $\Gamma^{-1}$ to the corresponding qubit to make it the identity. By the symplectic condition, there are an odd number of identity elements in the row. We may remove pairs of identity elements with a $\tfour$, similar to Lemma~\ref{lem:eliminate}, until there is only one left and the rest of the row is zero. The symplectic condition implies the column below the identity element is also zero, and we proceed to eliminate the rest of the tableau. Once the matrix part of the tableau is a permutation, we apply SWAP gates so that it becomes the identity and apply Pauli matrices to zero out the phase bits. 

We conclude that all egalitarian gates are in $\class{\tfour, \P, \Gamma}$. 
\end{proof}

%!TEX root = ../full_paper.tex
\newcommand{\circuitspacer}{3pt}

\section{Circuit Identities}
\label{sec:circuit_identities}
In this section, we give necessary tools to prove that a set of gates generates, in some sense, ``all that one could hope for.''  Formally, we wish to prove that the gate set generates a particular class in the classification lattice when it is contained in that class but fails to satisfy the invariants of all classes below it.  To this end, we give several useful circuit identities that will be used extensively in Section~\ref{sec:universal_construction}.  For instance, one can show that any circuit on two qubits can be reduced to an equivalent circuit containing at most one generalized CNOT gate (see Appendix~\ref{app:canonical_form}).   The following lemma gives only the aspect of that theorem that is necessary to the classification, that is, the ability to extract single-qubit Clifford operations from the composition of generalized CNOT gates.  

\begin{lemma}
\label{lem:coalesce_cnots}
Let $P, Q, R \in \P$, and let $\Gamma P \Gamma^\dag =  Q$ and $\Gamma Q \Gamma^\dag =  R$.  Then
\begin{itemize}
\item $C(P,Q)$ and $C(P,R)$ generate $\R_P$.
\item $C(P,P)$ and $C(P,Q)$ generate $\R_P$.
\item $C(P,P)$ and $C(Q,R)$ generate $\Gamma$.
\item $C(P,P)$ and $C(Q,Q)$ generate $\theta_{P+Q}$.
\end{itemize}
\end{lemma}
\begin{proof}
The first inclusion comes from the following identity:
\newcommand{\raiserise}{-15pt}
\begin{align*}
\Qcircuit @C=1em @R=1.5em {
& \gate{P} & \gate{P} & \qw  \\
& \gate{R} \qwx & \gate{Q} \qwx & \qw}
\horizontally
\raisebox{\raiserise}{=}
\horizontally
\Qcircuit @C=1em @R=1.5em {
& \gate{\R_{P}} & \gate{P} & \qw \\
& \qw & \gate{P} \qwx & \qw }
\horizontally
\raisebox{\raiserise}{$\xrightarrow[\text{rule}]{\text{ancilla}}$}
\horizontally
\raisebox{\raiserise+2pt}{\Qcircuit @C=1.5em @R=1.5em {
& \gate{\R_{P}}& \qw}}
\horizontally
\end{align*}
where we have used that $QR = iP$ is a Pauli. Conjugating the second qubit by $\Gamma$ in the diagram above, gives the second identity. 
For the third identity, we have 
\begin{align*}
\Qcircuit @C=1em @R=1.5em {
& \gate{P} & \gate{Q} & \gate{P} & \qw  \\
& \gate{P} \qwx & \gate{R} \qwx & \gate{P} \qwx & \qw}
\horizontally
\raisebox{\raiserise}{=}
\horizontally
\Qcircuit @C=1em @R=1.5em {
& \qswap  & \gate{\Gamma} & \qw \\
& \qswap \qwx  & \gate{\Gamma^\dag} & \qw}
\horizontally
\raisebox{\raiserise}{$\xrightarrow[\text{rule}]{\text{swap}}$}
\horizontally
\Qcircuit @C=1em @R=1.5em {
& \gate{\Gamma} & \qw \\
& \gate{\Gamma^\dag} & \qw}
\horizontally
\raisebox{\raiserise}{$\xrightarrow[\text{rule}]{\text{ancilla}}$}
\horizontally
\raisebox{\raiserise+2pt}{\Qcircuit @C=1.5em @R=1.5em {
& \gate{\Gamma} & \qw}}
\end{align*}
and for the final identity
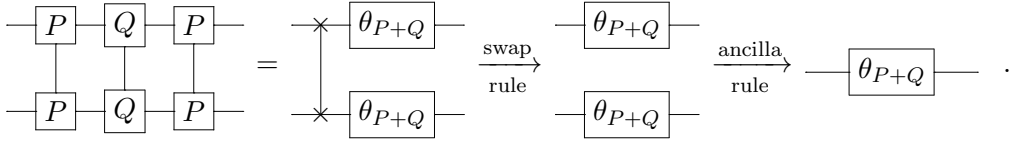
\begin{align*}
\Qcircuit @C=1em @R=1.5em {
& \gate{P}  & \gate{Q}  & \gate{P} & \qw  \\
& \gate{P} \qwx & \gate{Q} \qwx & \gate{P} \qwx & \qw}
\horizontally
\raisebox{\raiserise}{=}
\horizontally
\Qcircuit @C=1em @R=1.5em {
& \qswap  & \gate{\theta_{P+Q}}  & \qw  \\
& \qswap \qwx & \gate{\theta_{P+Q}} & \qw}
\horizontally
\raisebox{\raiserise}{$\xrightarrow[\text{rule}]{\text{swap}}$}
\horizontally
\Qcircuit @C=1em @R=1.5em {
& \gate{\theta_{P+Q}}  & \qw  \\
& \gate{\theta_{P+Q}}  & \qw}
\horizontally
\raisebox{\raiserise}{$\xrightarrow[\text{rule}]{\text{ancilla}}$}
\horizontally
\raisebox{\raiserise}{\Qcircuit @C=1.5em @R=1.5em {
& \gate{\theta_{P+Q}}  & \qw}\;\; .}
\end{align*}
\end{proof}

It might seem strange to reduce non-degenerate gates into less powerful single-qubit gates, but we will eventually see that the single-qubit generators are crucial.  Once we have shown that a particular set of gates generates all single-qubit operations, then that set of gates will generate the class of \emph{all} Clifford operations provided it contains any non-degenerate gate. All non-degenerate gates generate at least one Pauli, often the entire Pauli group, which is why some single-qubit classes do not appear as the single-qubit subgroup of a non-degenerate class.  For instance, consider the CNOT gate where the first qubit controls the second qubit.  If we let the first input be $\ket{1}$, then a Pauli $X$ operation is always applied to the second qubit.  Similarly, if we let the input to the second qubit be $\ket{-}$, then a Pauli $Z$ operation is always applied to the first qubit.  Under the ancilla rule, we now have Pauli $X$ and $Z$ operations, so we can generate $Y$ and the entire Pauli group.  Clearly, the same is true for any heterogeneous CNOT gate.  However, surprisingly, the following lemma shows that even the $\tfour$ gate suffices to generate the entire Pauli group.

\begin{lemma}
\label{lem:t4_generates_paulis}
$\tfour$ generates the Pauli group.
\end{lemma}
\begin{proof}
Consider the following two circuits:
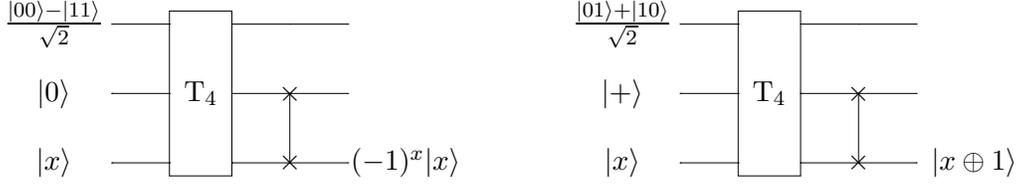
\begin{center}
\begin{minipage}{.45\linewidth}
\centering
$\Qcircuit @C=2em @R=1.5em {
\frac{\ket{00} - \ket{11}}{\sqrt{2}} & & \multigate{2}{\tfour} & \qw & \qw \\
\ket{0} & & \ghost{\tfour} & \qswap & \qw  & \\
\ket{x} & & \ghost{\tfour} & \qswap \qwx & \qw & (-1)^x \ket{x}}$
\end{minipage}
\begin{minipage}{.45\linewidth}
\centering
$\Qcircuit @C=2em @R=1.5em {
\frac{\ket{01} + \ket{10}}{\sqrt{2}} & & \multigate{2}{\tfour} & \qw & \qw \\
\ket{+} & & \ghost{\tfour} & \qswap & \qw  & \\
\ket{x} & & \ghost{\tfour} & \qswap \qwx & \qw & \ket{x \oplus 1}}$
\end{minipage}
\end{center}
Under the ancilla rule, the first generates a Pauli $Z$ operation while the second generates a Pauli $X$, from which we can clearly generate the Pauli group.
% \Since $\tfour$ acts identically in the $X$-basis (TODO, prove), it is clear that the technique can used to generate an $X$ gate. 
\end{proof}

%There is another way to view the identity of Lemma~\ref{lem:t4_generates_paulis} which will be useful later.  Since $\tfour$ is an affine gate over the computational basis states, $\tfour = C_1 C_2 \ldots C_n$ where each $C_i$ is a CNOT gate.  Furthermore, $\tfour$ and CNOT are their own inverses, so $\tfour = C_n C_{n-1} \ldots C_1$.  Finally, because $\tfour$ is symmetric when represented as a $4 \times 4$ matrix over $\ftwo$, $\tfour = C_1^T C_2^T \ldots C_n^T$.  Notice that $C_i^T$ just represents the CNOT gate $C_i$ where the control and target qubits are swapped.  Therefore, leveraging the well-known equivalence $\theta^{\otimes 2}_{X+Z} \circ \CNOT \circ \theta^{\otimes 2}_{X+Z}  = \SWAP \circ \CNOT \circ \SWAP$ we arrive at the following consequence:\footnote{Recall that $\theta_{X+Z}$ is commonly known as the Hadamard gate.}
%$$\theta_{X+Z}^{\otimes 4} \tfour \theta_{X+Z}^{\otimes 4} = \tfour.$$
%Similarly, by straightforward calculation, we get
%$$\R_{Z}^{\otimes 4} \tfour \R_{Z}^{\otimes 4} = \tfour.$$
%It is now easy to extend old circuit identities into new ones.  For instance, conjugating the first circuit in the proof of Lemma~\ref{lem:t4_generates_paulis} by $\theta_{X+Z}^{\otimes 4}$ (which does not change the circuit because of the above observations) and pushing the $\theta_{X+Z}$ gates into the inputs, yields the second circuit.  This technique is in fact very general and is used in the proof of the lemma below.

\begin{lemma}
\label{lem:tfour_cpp_to_rp}
$\tfour$ and $C(P,P)$ generate $\R_P$.
\end{lemma}
\begin{proof}
Figure \ref{fig:t4_plus_cpp_to_rp} shows how to generate $\R_Z$ with $C(Z,Z)$. To generate $\R_X$ with $C(X,X)$, we simply appeal to the symmetry of $\tfour$.  Concretely, consider conjugating the entire circuit in Figure~\ref{fig:t4_plus_cpp_to_rp} by $\Gamma^{\otimes 4}$. We must now check: the circuit generates $\R_X$; and it is in the class $\class{\tfour, C(X,X)}$.  For the former, note that we can still apply the ancilla rule on the first three qubits (i.e., by multiplying the original ancillas by $\Gamma$ on each qubit). Therefore, the transformation on the final qubit is $\Gamma \R_Z \Gamma^\dagger = \R_X$.

To check containment in the class, we simply conjugate all the gates in the circuit by $\Gamma^{\otimes 4}$.  SWAP is unchanged, and $C(Z,Z)$ is mapped to $C(X,X)$.  Finally, recall that $\tfour$ is egalitarian, so
$$
\Tstar(\Gamma^{\otimes 4} \tfour (\Gamma^\dagger)^{\otimes 4}) = \Tstar(\tfour).
$$
That, it is mapped to is itself followed by layer of Pauli operations (technically, we have the identity $\Gamma^{\otimes 4} \tfour (\Gamma^\dagger)^{\otimes 4} = \tfour X^{\otimes 4}$).  Since $\tfour$ generates the Pauli group by Lemma~\ref{lem:t4_generates_paulis}, the entire conjugated circuit is in $\class{\tfour, C(X,X)}$.  Repeating the entire argument for $C(Y,Y)$ and $\R_Y$ completes the proof.
\end{proof}

\begin{figure}
\centering
\mbox{
\Qcircuit @C=1.5em @R=1em {
\raisebox{-4.1em}{\begin{tikzpicture}\draw [decorate,decoration={brace}](0,0)--(0,.82) node[midway, xshift=-2.2em]{$\frac{\ket{00} + i\ket{11}}{\sqrt{2}}$};\end{tikzpicture}\hspace{2em}} & & \multigate{3}{\tfour} & \qw & \qw & \qw & \\
 & & \ghost{\tfour} & \gate{Z} & \qw & \qw & \\
\ket{0} & & \ghost{\tfour} &  \gate{Z} \qwx & \qswap & \qw  & \\
\ket{x} & & \ghost{\tfour} &  \qw & \qswap \qwx & \qw & i^x \ket{x}
}}
\caption{Generating $\R_Z$ with $\tfour$ and $C(Z,Z)$.}
\label{fig:t4_plus_cpp_to_rp}
\end{figure}
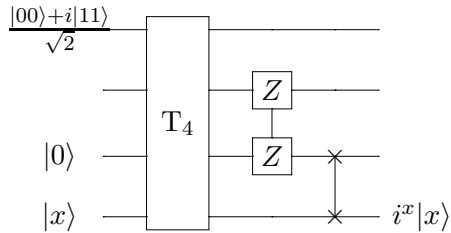

The following lemmas make precise our working assumption that single-qubit gates can significantly bolster the power of non-degenerate gate sets.
\begin{lemma}
\label{lem:cpq_plus_to_all}
Suppose we have any $C(P,Q)$ gate with any single-qubit gate $G$ that does not preserve the $P$-basis and any single-qubit gate $H$ that does not preserve the $Q$-basis.  Then $\class{C(P,Q), G, H} = \ALL$.
\end{lemma}
\begin{proof}
We will prove that the class $\class{C(P,Q), G, H}$ contains all single-qubit gates.  Then, to prove that the class generates all Clifford operations, it is be sufficient to show that it contains a CNOT gate.  However, since all generalized CNOT gates are conjugates of each other, this is immediate.

First suppose $P=Q$.  Since $G$ does not preserve $P$-basis, we can use $G$ to create a $C(R,R)$ gate where $R \neq P$.  By Lemma~\ref{lem:coalesce_cnots}, we can generate a $\theta_{P+R}$ gate.  Conjugating $C(P,P)$ by $\theta_{P+R}$ on the second qubit yields a $C(P,R)$ gate.  Once again leveraging Lemma~\ref{lem:coalesce_cnots},  $C(P,R)$ and $C(P,P)$ generate an $\R_P$ gate.  Referring to the single-qubit lattice (see Figure~\ref{fig:degeneratelattice}), we see that the class $\class{\P, \theta_{P+R}, \R_P}$ contains all single-qubit gates.

Now suppose that $P \neq Q$.  Once again, since $G$ does not preserve $P$-basis, we can use $G$ to create a $C(R, Q)$ gate.  If $R = Q$, then by the logic above, we can use $H$ to generate all single-qubit gates, so suppose $R \neq Q$.  By Lemma~\ref{lem:coalesce_cnots}, we can use $C(P,Q)$ and $C(R,Q)$ to generate an $\R_Q$ gate.  Conjugating both $C(P,Q)$ and $C(R,Q)$ by $H$ appropriately, gives a $C(P,S)$ and $C(R,S)$ for some $S \neq Q$, which we can once again generate an $\R_S$ gate.  Referring to the single-qubit lattice, we see that the class $\class{\P, \R_S, \R_Q}$ contains all single-qubit gates.
\end{proof}

\begin{lemma}
\label{lem:t4_plus_single_qubits_ftw}
$\tfour$ with the class of all single-qubit gates generates $\ALL$.
\end{lemma}
\begin{proof}
It is well known that $\CNOT$, $\theta_{X+Z}$, and $\R_Z$ generate all Clifford circuits.  Therefore, it will be sufficient to show that $\tfour$ plus all single-qubit gates generate $\CNOT$.  Under the ancilla rule, it is clear by Figure~\ref{fig:t4_plus_s_generates_czz} that $\tfour$ and $\R_Z$ suffice to generate $C(Z,Z)$.  Conjugating one qubit of $C(Z,Z)$ by $\theta_{X+Z}$ yields a $C(Z,X) = \CNOT$ gate, completing the proof.
\end{proof}

\begin{figure}
\begin{align*}
\Qcircuit @C=1em @R=.7em @!R {
 & & \multigate{3}{\tfour} & \gate{\R_Z^\dag} & \multigate{3}{\tfour} & \qw & \qw &&&  \qw & \qw & \qw & \qw\\
 & & \ghost{\tfour} & \qw & \ghost{\tfour} & \gate{\R_Z} & \qw &&& \gate{Z} & \gate{Z} \qwx[2] & \qw & \qw \\
 & & \ghost{\tfour} & \qw & \ghost{\tfour} & \gate{\R_Z} & \qw & \raisebox{2.3em}{$=$} && \gate{Z} \qwx & \qw & \gate{Z} & \qw \\
 & & \ghost{\tfour} & \qw & \ghost{\tfour} & \gate{\R_Z} & \qw &&& \qw & \gate{Z} & \gate{Z} \qwx & \qw}
\end{align*}
\caption{Generating $C(Z,Z)$ with $\tfour$ and $\R_Z$.}
\label{fig:t4_plus_s_generates_czz}
\end{figure}
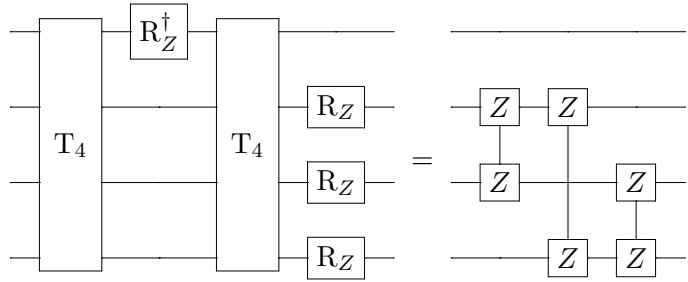

%!TEX root = ../full_paper.tex

\section{Universal Construction}
\label{sec:universal_construction}

Suppose $G$ is an $n$-qubit Clifford gate. It turns out there is a single circuit $\mathfrak{C}(G)$, the \emph{universal construction}, which can help us extract useful generators (e.g., single-qubit gates, generalized CNOTs, etc.) from $G$. Specifically, the circuit $\mathfrak{C}(G)$ (shown in Figure~\ref{fig:universal_construction}) applies $G$ to qubits $2$ through $n+1$, swaps qubits $1$ and $2$, then applies $G^\dag$ to qubits $2$ through $n+1$.

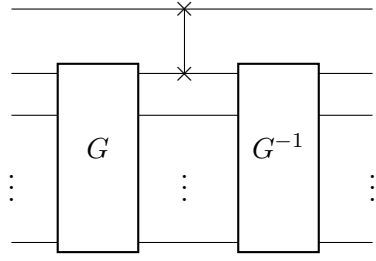
\begin{figure}
\begin{center}
\begin{tikzpicture}[>=latex]
\tikzstyle{gate}=[rectangle, thick, fill=white, minimum size=0.5cm, draw, align=center]
\tikzstyle{dots}=[rectangle, fill=white, minimum size=0.5cm, align=center]
\newcommand{\gatewidth}{.8cm}
\newcommand{\wiresep}{.3cm}
\matrix[row sep=.3cm,column sep=.5cm] {
\coordinate(start1); \node[minimum width=\wiresep]{}; & & \node(TOP)[inner sep=0pt]{$\times$}; & & \node[minimum width=\wiresep]{}; \coordinate(end1); \\ [.3cm]
\coordinate(start2); & \node(GT)[minimum width=\gatewidth]{}; & \node(BOT)[inner sep=0pt]{$\times$}; & \node(GIT)[minimum width=\gatewidth]{}; & \coordinate(end2); \\
\coordinate(start3); & & & \node(GINV)[minimum width=1cm]{};& \coordinate(end3); \\
 \node[dots, xshift=.5cm]{$\vdots$}; & & \node[dots]{$\vdots$};& & \node[dots, xshift=-.5cm]{$\vdots$};\\
\coordinate(start4); & \node(GB)[minimum width=\gatewidth]{}; & &\node(GIB)[minimum width=\gatewidth]{};  & \coordinate(end4); \\
};
\node[gate, fit={(GT.west) (GB.east)}] {$G$};
\node[gate, fit={(GIT.west) (GIB.east)}] {$G^{\dag}$};
% EDGES
\begin{pgfonlayer}{background}
\path[draw] 
(start1) edge (end1)
(start2) edge (end2)
(start3) edge (end3)
(start4) edge (end4)
(TOP.center) edge (BOT.center);
\end{pgfonlayer}
\end{tikzpicture}
\end{center}
\caption{Universal Construction $\mathfrak{C}(G)$}
\label{fig:universal_construction}
\end{figure}

The intuition is that since we apply $G$ and $G^{\dag}$ back-to-back on all but one of the same qubits, they ``mostly'' cancel and the result depends on the influence of the swapped qubit. In fact, we will show that the universal construction is equivalent to a circuit derived from the column of the tableau associated with that qubit. If that part of the tableau violates a particular invariant, then this circuit contains some gate \emph{also} violating that invariant, and in the next section we show how to extract it. For now, we focus on proving our claims about the universal construction.

Our main tool for analyzing $\mathfrak{C}(G)$ is a novel canonical form for Clifford circuits. For all $k \geq 1$, let $O_{2k-1}$ be a $(2k-1)$-qubit gate of the form
$$
\Tstar(O_{2k-1}) = 
\begin{pmatrix} 
I & I & I & \cdots & I \\
I & \alpha & \alpha^{*} & \cdots & \alpha^{*} \\
I & \alpha^{*} & \alpha & & \alpha^{*} \\
\vdots & \vdots & & \ddots & \vdots \\
I & \alpha^{*} & \alpha^{*} & \cdots & \alpha
\end{pmatrix}
$$
where $\alpha = (\begin{smallmatrix} 0 & 1 \\ 1 & 1 \end{smallmatrix}), I = (\begin{smallmatrix} 1 & 0 \\ 0 & 1 \end{smallmatrix}) \in \ring$. That is, $I$ in the first row and column, otherwise $\alpha$ on the diagonal and $\alpha^{*}$ off-diagonal. We note that the three-qubit gate $O_{3}$ generates $\class{\tfour, \Gamma, \P}$ (see discussion in Appendix~\ref{app:gta} for a similar gate up to single-qubit gates).

\newcommand{\nogate}[1]{*+<1em,.9em>{#1} \qw}
\newcommand{\nogatenowire}[1]{*+<1em,0em>{#1}}
\begin{figure}
    \centering
    \mbox{
    \Qcircuit @C=1em @R=1em {
    & \gate{\G(r_1)} & \multigate{3}{O_{2k-1}} & \gate{r_k \vphantom{r_k^{*}}} \qwx[4] & \nogate{\cdots} & \gate{r_{n} \vphantom{r_n^{*}}} \qwx[6] & \gate{P} & \qw \\
    & \gate{\G(r_2)} & \ghost{O_{2k-1}} & \qw & \qw & \qw & \multigate{5}{\;\;\;D\;\;\;} & \qw \\
    & \nogatenowire{\vdots} & & & \\
    & \gate{\G(r_{2k-1})} & \ghost{O_{2k-1}} & \qw & \qw & \qw & \ghost{\;\;\;D\;\;\;} & \qw \\
    & \qw & \qw & \gate{r_k^{*}} & \qw & \qw & \ghost{\;\;\;D\;\;\;} & \qw \\
    & &  &  & \nogatenowire{\ddots} &  &  \\
    & \qw & \qw & \qw & \qw & \gate{r_{n}^{*}} & \ghost{\;\;\;D\;\;\;} & \qw  \\
    }
    }
    \caption{A diagram of the decomposition in Lemma~\ref{lem:decomp}. For convenience, $S = \{1, \ldots, 2k-1\}$ and $i = 1$, so $O_{2k-1}$ is on the first $2k-1$ qubits and no swap is necessary.}
    \label{fig:decomposition}
\end{figure}

\begin{lemma}
    \label{lem:decomp}
    Let $G$ be an arbitrary $n$-qubit Clifford gate and suppose there are $2k-1$ invertible elements in the first column of $G$. Then $G$ can be decomposed into the following sequence of gates (see Figure~\ref{fig:decomposition}):
    \begin{enumerate}
        \item single-qubit gates on a subset $S$ of the qubits,
        \item an $O_{2k-1}$ gate on the same subset $S$,
        \item a generalized CNOT from a fixed $i \in S$ to each qubit outside $S$,
        \item a $\SWAP$ on qubit $1$ and qubit $i$ (or no swap if $i = 1$),
        \item a Pauli $P$ on qubit $1$,
        \item an $(n-1)$-qubit Clifford gate $D$ on qubits $2, \ldots, n$. 
    \end{enumerate}
\end{lemma}
\begin{proof}
    Our goal is to show that $G$ decomposes into these layers of gates:
    $$
    G = (P \otimes D) \circ \mathrm{SWAP}(1,i) \circ \left( \prod \CNOT(r_j) \right) \circ O_{2k-1} \circ \left( \bigotimes \G(r_j) \right).
    $$
    If we look at the matrix part only and isolate $P \otimes D$, we have
    $$
    \Tstar(\mathrm{SWAP}(1,i)) \Tstar \left( \prod \CNOT(r_j) \right)  \Tstar(O_{2k-1}^{\dag})  \Tstar\left( \bigotimes \G(r_j)^{\dag} \right)  \Tstar(G) = \Tstar(P \otimes D).
    $$
    To prove the decomposition, we start with the first column of $\Tstar(G)$ and show that this sequence of gates/matrices simplifies it. This process looks a bit like one step of Gaussian elimination: we pick a pivot row $i$, perform operations (i.e., gates) to zero out all other entries of the column, and swap the pivot row with the first row. Let $(r_1, \ldots, r_n)$ be the first column of $\Tstar(G)$. Take $S = \{ j \in \{ 1, \ldots, n \} : \text{$r_j$ is invertible in $\ring$} \}$ to be the indices of the invertible elements, and let the \emph{pivot} $i \in S$ be an arbitrary index (e.g., take the minimum element of $S$ as the canonical choice). We note that $S$ must have an odd number of elements because the tableau is symplectic. 
    
    For each invertible $r_j \in \ring$, there exists a gate $\G(r_j)$ (unique up to Paulis/phase) which maps $I$ to $r_j$. Thus, when we apply $\G(r_j)^{\dag}$ to qubit $j$, it maps $r_j$  to $I$. That is, the layer of single-qubit gates multiplying the rows such that the invertible elements of the column are sent to $I$. Next, we apply $O_{2k-1}$ to the invertible qubits $S$ to further simplify the column. Notice that $O_{2k-1}$ is not completely symmetric---the first qubit is special, but all the others are interchangeable. Let the pivot row, $i$, be the special qubit when we apply $O_{2k-1}$. We leave it as an exercise to check that $O_{2k-1}$ maps $(I,0,\ldots,0)$ to $(I,I,\ldots,I)$, and therefore $O_{2k-1}^{\dag}$ will simplify the column to $(I,0,\ldots,0)$ on indices $S$. 
    
    Next, for each $r_j \notin S$, there exists a generalized CNOT gate $\CNOT(r_j)$ that maps $(I, r_j)$ to $(I, 0)$. Since the $i$th entry of the column is now $I$, and the $j$th entry is $r_j$, applying $\CNOT(r_j)$ zeros out another entry of the column. After all generalized CNOTs have been applied, the column vector is all $0$'s except for an $I$ in the pivot row. We apply a SWAP on qubit $1$ and $i$ to move this pivot row to the top, so the column vector must be $(I,0,\ldots,0)$. By the symplectic conditions, one can check that the first row must also be $(I,0,\ldots,0)$, and hence the tableau is of the form
    $$
    \left(\begin{array}{c|ccc} 
    I      & 0 & \cdots & 0 \\ \hline
    0      &   &        & \\
    \vdots &   &   A    &  \\
    0      &   &        & 
    \end{array}\right)
    $$
    In other words, we are left with a Pauli on qubit $1$ (since there may be sign bits) and a Clifford gate on qubits $2$ through $n$. We choose $P$ and $D$ to be these components of the tableau, and that completes the decomposition. 
\end{proof}

\begin{corollary}
\label{cor:universal_construction_form}
    Let $G$ be a Clifford gate on $n$ qubits where the first column of the tableau is $(a_2,\ldots,a_{2k},b_1,\ldots,b_\ell)$ for invertible $a_2, \ldots, a_{2k} \in \ring$ and noninvertible $b_1, \ldots, b_\ell \in \ring$. Then $\mathfrak{C}(G)$ is as shown in Figure~\ref{fig:universal_decomposition}: a $\tgate_{2k}$ on qubits $1$ through $2k$, conjugated by single-qubit gates (except on qubit $1$) derived from $a_2, \ldots, a_{2k}$, an array of generalized CNOT gates derived from $b_1, \ldots, b_{\ell}$ on either side between qubit $1$ and qubits $2k+1$ through $n+1$, and conjugated by a Pauli $P$ on qubit $1$.  
\end{corollary}
\begin{proof}
    Apply the decomposition (Lemma~\ref{lem:decomp}) to $G$. This decomposition ends with an $(n-1)$-qubit gate $D$; in the universal construction, $D$ appears innermost in the circuit, next to the SWAP, acting on qubits $3$ through $n+1$. Since $D$ and the SWAP act on disjoint qubits, they commute, so $D$ cancels with $D^{\dag}$ on the other side (from the inverted decomposition for $G^{\dag}$). 
    
    After cancelling $D$ with $D^{\dag}$, the next innermost gates are the Pauli $P$ and generalized CNOT gates acting on qubit $2$. Let us push these gates through the SWAP (from both sides) so that they all act on qubit $1$ instead of $2$, i.e., as they appear in Figure~\ref{fig:universal_decomposition}. At this point the SWAP is flanked by $O_{2k-1}$ and $O_{2k-1}^{\dag}$. It is a straightforward calculation to check that $\mathfrak{C}(O_{2k-1}) = \tgate_{2k}$, so we replace SWAP, $O_{2k-1}$, and $O_{2k-1}^{\dag}$ with $\tgate_{2k}$. Finally, the decomposition has single-qubit gates $\G(a_2), \ldots, \G(a_{2k})$, which are matched by their inverses $\G(a_2)^{\dag}, \ldots, \G(a_{2k})^{\dag}$ in $G^{\dag}$, i.e., $\tgate_{2k}$ is conjugated by the appropriate single-qubit gates on qubits $2$ through $2k$.  
\end{proof}
\newcommand{\dotshift}{3pt}
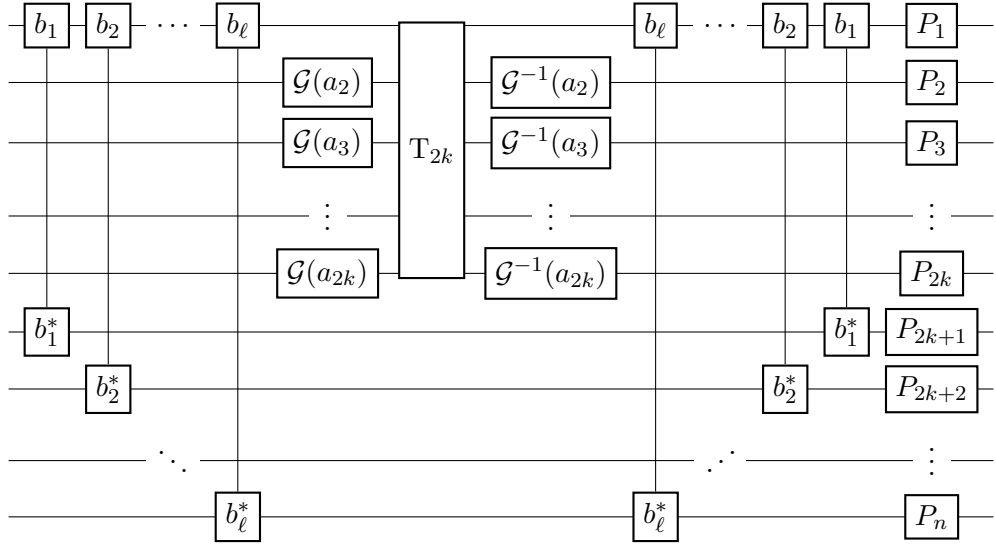
\begin{figure}
\begin{center}
\begin{tikzpicture}[>=latex]
\tikzstyle{gate}=[rectangle, thick, fill=white, minimum size=0.5cm, draw, align=center]
\tikzstyle{dots}=[rectangle, fill=white, minimum size=0.5cm, align=center]
\matrix[row sep=.1cm,column sep=.2cm] {
\coordinate(start1); & \node[gate]{$P$}; &  \node(LTB1)[gate]{$b_1$}; & \node(LTB2)[gate]{$b_2$}; & \node [dots] {$\ldots$}; & \node (LTBL) [gate]{$b_\ell$}; & &   & & \node (RTBL) [gate]{$b_\ell$};  & \node [dots] {$\ldots$}; & \node (RTB2) [gate] {$b_2$}; & \node (RTB1) [gate]{$b_1$}; & \node[gate]{$P$};& \coordinate(end1); \\
\coordinate(start2); & & & & & & \node[gate]{$\G(a_2)$}; &    & \node[gate]{$\G(a_2)^{\dag}$}; & & & & & & \coordinate(end2); \\
\coordinate(start3); & & & & & & \node[gate]{$\G(a_3)$}; &  \node(T)[minimum width=1cm]{};   & \node[gate]{$\G(a_3)^{\dag}$}; & & & & &  &\coordinate(end3); \\[-.1cm]
\coordinate(start4); & & & & & & \node[dots, yshift=\dotshift]{$\vdots$}; &     & \node[dots, yshift=\dotshift]{$\vdots$};  & & & & & & \coordinate(end4); \\
\coordinate(start5); & & & & & & \node[gate]{$\G(a_{2k})$}; &    & \node[gate]{$\G(a_{2k})^\dag$};  & & & & & &\coordinate(end5); \\
\coordinate(start6); & & \node(LBB1)[gate]{$b_1^*$};  & & & & &    & & & & & \node(RBB1)[gate]{$b_1^*$}; & &\coordinate(end6); \\
\coordinate(start7); & & &  \node(LBB2)[gate]{$b_2^*$}; & & & &   & & & & \node(RBB2)[gate]{$b_2^*$}; & & &\coordinate(end7); \\
\coordinate(start8); & & & & \node[dots, yshift=\dotshift]{$\ddots$}; & & &   & & & \node[dots, yshift=\dotshift]{$\iddots$}; & & & & \coordinate(end8); \\
\coordinate(start9); & & & & & \node(LBBL)[gate]{$b_\ell^*$}; & &    & & \node(RBBL)[gate]{$b_\ell^*$}; & & & & &\coordinate(end9); \\
};
% EDGES
\begin{pgfonlayer}{background}
\path[draw] 
(start1) edge (end1)
(start2) edge (end2)
(start3) edge (end3)
(start4) edge (end4)
(start5) edge (end5)
(start6) edge (end6)
(start7) edge (end7)
(start8) edge (end8)
(start9) edge (end9)
(LTB1) edge (LBB1)
(LTB2) edge (LBB2)
(LTBL) edge (LBBL)
(RTBL) edge (RBBL)
(RTB2) edge (RBB2)
(RTB1) edge (RBB1);
\node[gate, minimum height=3.7cm, yshift=-.1cm] at (T){$\operatorname{T}_{2k}$};
\end{pgfonlayer}
\end{tikzpicture}
\end{center}
\caption{Decomposition of $\mathfrak{C}(G)$.}
\label{fig:universal_decomposition}
\end{figure}

We are finally ready to prove the main theorem of this section.

\begin{theorem}
\label{thm:gate_extraction}
Let $G$ be a Clifford gate on $n$ qubits.  Furthermore, let  $(a_2, a_3, \ldots, a_{2k}, b_1, \ldots, b_\ell) \in \ring^{n}$ be some column of $\Tstar(G)$, where each $a_i$ is invertible and each $b_i$ is noninvertible. Then $G$ generates a gate $\G_i$ such that $\Tstar(\G_i) = a_i$ for each $i \in \{2, \ldots, 2k\}$, generates $\CNOT(b_j)$ for all $j \in \{1, \ldots, \ell\}$, and generates a $\operatorname{T}_{2k}$ gate.
\end{theorem}
\begin{proof}
From Corollary~\ref{cor:universal_construction_form}, the universal construction $\mathfrak{C}(G)$ can be decomposed as shown in Figure~\ref{fig:universal_decomposition}.  The proof will proceed in the following manner.  Starting with the decomposition of  $\mathfrak{C}(G)$, we show that it generates some elementary gate.  We then use that gate to further simplify the original decomposition of $\mathfrak{C}(G)$ until we have generated all the gates specified in the theorem.  

First notice that for each input $i \in \{2k+1, \ldots, n\}$, there exists a single-qubit Clifford state $\ket{b_i}$ that eliminates the generalized $\CNOT(b_i)$ gate (e.g., $\ket{0}$ on the control of a CNOT gate). Therefore, we can generate the following gate:
%\begin{figure}
\begin{center}
\begin{tikzpicture}[>=latex]
\tikzstyle{gate}=[rectangle, thick, fill=white, minimum size=0.5cm, draw, align=center]
\tikzstyle{dots}=[rectangle, fill=white, minimum size=0.5cm, align=center]
\matrix[row sep=.1cm,column sep=.2cm] {
\coordinate(start1); & \node[gate]{$P$}; &   & \node[gate]{$P$}; & \coordinate(end1); \\
\coordinate(start2);  & \node[gate]{$\G(a_2)$}; &    & \node[gate]{$\G(a_2)^\dag$}; & \coordinate(end2); \\
\coordinate(start3);  & \node[gate]{$\G(a_3)$}; &  \node(T)[minimum width=1cm]{};   & \node[gate]{$\G(a_3)^{\dag}$}; & \coordinate(end3); \\[-.1cm]
\coordinate(start4);  & \node[dots, yshift=\dotshift]{$\vdots$}; &     & \node[dots, yshift=\dotshift]{$\vdots$};  & \coordinate(end4); \\
\coordinate(start5);  & \node[gate]{$\G(a_{2k})$}; &    & \node[gate]{$\G(a_{2k})^{\dag}$};  &\coordinate(end5); \\
};
% EDGES
\begin{pgfonlayer}{background}
\path[draw] 
(start1) edge (end1)
(start2) edge (end2)
(start3) edge (end3)
(start4) edge (end4)
(start5) edge (end5);
\node[gate, minimum height=3.7cm, yshift=-.1cm] at (T){$\operatorname{T}_{2k}$};
\end{pgfonlayer}
\end{tikzpicture}
\end{center}
%\caption{Applying ancilla rule to qubits $2k+1$ through $n$ of $\mathfrak{C}(G)$.}
%\label{fig:easy_reduce}
%\end{figure}
Now let $\ket{\phi} = \frac{\ket{00} + \ket{11}}{\sqrt{2}}$ be the Bell state on two qubits. Notice that we can use $\ket{\phi}$ as an ancilla to remove two bits from $\tgate_{2\ell}$ (i.e., leaving a $\tgate_{2\ell-2}$). Therefore, the state $\G(a_i)^\dag \otimes \G(a_j)^\dag \ket{\phi}$ on bits $i$ and $j$ (for $i,j \neq 1$) is fixed by the circuit and simplifies the $\tgate_{2k}$ gate.  We can iterate this procedure on the circuit above until it has been reduced to just two qubits. In particular, the $\tgate_{2k}$ gate in the middle is now a $\tgate_2$, otherwise known as the $\SWAP$ gate. This remaining circuit is
\begin{center}
$\Qcircuit @C=1.8em @R=1em {
& \gate{P} & \qswap & \gate{P} & \qw  \\
& \gate{\G(a_i)} & \qswap \qwx & \gate{\G(a_i)^{\dag}} & \qw}$
\end{center}
Commute the first layer of gates through the swap and apply the swap rule.  This leaves a tensor product of single-qubit gates, i.e., $P \circ \G(a_i)$ and its inverse, each of which we can isolate using the ancilla rule.

Now let us repeat the procedure above starting from $\mathfrak{C}(G)$, but stop short of applying the ancilla rule to qubit  $2k+j$.  The result is the first circuit depicted below, which is then simplified by swapping the first two qubits, and adding gates $P \circ \G(a_i), \G(a_i)^{\dag} \circ P$:
$$
\vcenter{\Qcircuit @C=1em @R=1em {
& \gate{P} & \gate{b_j} & \qw & \qswap & \qw & \gate{b_j} & \gate{P} & \qw  \\
& \qw & \qw & \gate{\G(a_i)} & \qswap \qwx & \gate{\G(a_i)^{\dag}}  & \qw & \qw &  \qw \\
& \qw & \gate{b_j^*} \qwx[-2] & \qw & \qw & \qw & \gate{b_j^*} \qwx[-2] & \qw & \qw 
}} \;\;\; \rightarrow \;\;\;
\vcenter{\Qcircuit @C=1em @R=1em {
& \qw & \qw & \gate{\G(a_i)}  & \gate{b_j} & \gate{\G(a_i)^{\dag}} & \qw  \\
& \gate{P} & \gate{b_j} & \gate{P} & \qw & \qw & \qw \\
& \qw & \gate{b_j^*} \qwx[-1] & \qw & \gate{b_j^*} \qwx[-2]  & \qw & \qw 
}}$$

Notice that for the circuit on the right, we can apply the ancilla rule using the state $\G(a_i)^{\dag}\ket{b_j}$ on the topmost qubit, leaving $\CNOT(b_j)$ conjugated by $P \otimes I$.
% Notice that the topmost qubit is stabilized by $\G(a_i)^{\dag}\ket{b_j}$, from which we can see that the ancilla rule immediately generates $\CNOT(b_j)$ conjugated by $P \otimes I$. 
Let us assume that $\CNOT(b_j)$ is a $C(Q,R)$ gate for some Paulis $Q$ and $R$.  If $Q = P$, then conjugation by $P \otimes I$ does nothing, and we obtain a $\CNOT(b_j)$ gate. Otherwise, the conjugation results in the gate $C(Q,R) \circ (I \otimes R)$ from which we can generate a Pauli $R$ using the ancilla rule and the state stabilized by $Q$ on the first qubit.  In either case, we eventually generate a $\CNOT(b_j)$ gate. 

Finally, to generate the $\tgate_{2k}$ we exploit the identity
$$\tgate_{2k} (P \otimes I^{\otimes 2k-1}) \tgate_{2k} = I \otimes P^{\otimes 2k-1},$$
which holds up to a global phase.  We have the following chain of consequences:

\begin{center}
\hspace{-20pt}
\begin{tikzpicture}[>=latex,scale=0.8, every node/.style={scale=0.8}]
\tikzstyle{gate}=[rectangle, thick, fill=white, minimum size=0.5cm, draw, align=center]
\tikzstyle{dots}=[rectangle, fill=white, minimum size=0.4cm, align=center]
\tikzstyle{phantom}=[minimum height=.65cm,minimum width=1cm]
\matrix[row sep=.1cm,column sep=.2cm] {
\coordinate(start11); & \node[gate]{$P$}; &  & \node[gate]{$P$}; & \coordinate(end11); & &  
\coordinate(start21); & & &  & & & \coordinate(end21); & &
\coordinate(start31); & & \coordinate(end31); \\
\coordinate(start12);  & \node[gate]{$\G(a_2)$}; & & \node[gate]{$\G(a_2)^\dag$}; & \coordinate(end12); & &
\coordinate(start22);  & \node[gate]{$\G(a_2)$}; &  \node[gate]{$P$}; & & \node[gate]{$P$}; & \node[gate]{$\G(a_2)^\dag$}; & \coordinate(end22); & & 
\coordinate(start32);  & & \coordinate(end32); \\
\coordinate(start13);  & \node[gate]{$\G(a_3)$}; &  \node(T1)[phantom]{};   & \node[gate]{$\G(a_3)^\dag$}; & \coordinate(end13); & \node{$\rightarrow$}; & 
\coordinate(start23);  & \node[gate]{$\G(a_3)$}; &  \node[gate]{$P$}; & \node(T2)[phantom]{};   & \node[gate]{$P$}; & \node[gate]{$\G(a_3)^\dag$}; & \coordinate(end23); & \node{$\rightarrow$}; &
\coordinate(start33);  & \node(T3)[phantom]{};   &\coordinate(end33); \\[-.1cm]
\coordinate(start14);  & \node[dots, yshift=\dotshift]{$\vdots$}; & & \node[dots, yshift=\dotshift]{$\vdots$};  & \coordinate(end14); & &
\coordinate(start24);  & \node[dots, yshift=\dotshift]{$\vdots$}; &  \node[dots, yshift=\dotshift]{$\vdots$}; & & \node[dots, yshift=\dotshift]{$\vdots$}; & \node[dots, yshift=\dotshift]{$\vdots$}; & \coordinate(end24); & &
\coordinate(start34);  &  & \coordinate(end34); \\
\coordinate(start15);  & \node[gate]{$\G(a_{2k})$}; & & \node[gate]{$\G(a_{2k})^\dag$}; & \coordinate(end15); & & 
\coordinate(start25);  & \node[gate]{$\G(a_{2k})$}; &  \node[gate]{$P$}; & & \node[gate]{$P$}; & \node[gate]{$\G(a_{2k})^\dag$}; & \coordinate(end25); & &
\coordinate(start35);  & &\coordinate(end35); \\
};
% EDGES
\begin{pgfonlayer}{background}
\path[draw] 
(start11) edge (end11)
(start12) edge (end12)
(start13) edge (end13)
(start14) edge (end14)
(start15) edge (end15)
(start21) edge (end21)
(start22) edge (end22)
(start23) edge (end23)
(start24) edge (end24)
(start25) edge (end25)
(start31) edge (end31)
(start32) edge (end32)
(start33) edge (end33)
(start34) edge (end34)
(start35) edge (end35);
\node[gate, minimum height=3.7cm, yshift=-.05cm] at (T1){$\operatorname{T}_{2k}$};
\node[gate, minimum height=3.7cm, yshift=-.05cm] at (T2){$\operatorname{T}_{2k}$};
\node[gate, minimum height=3.7cm, yshift=-.05cm] at (T3){$\operatorname{T}_{2k}$};
\end{pgfonlayer}
\end{tikzpicture}
\end{center}
where the last implication follows by applying the gates $P \circ \G(a_i)$ and $\G(a_i)^\dag \circ P$, which we generated previously.
\end{proof}

%!TEX root = ../full_paper.tex

\section{Completing the Classification}
\label{sec:completing}
The final step in the classification is to demonstrate that the classes we have defined are in fact the only classes that exist. First, we give a simple consequence of Theorem~\ref{thm:gate_extraction} that will make it easier to talk about gate sets that violate some invariant.

\begin{lemma}
\label{lem:invariant_to_gate}
Given that a gate set $G$ violates some invariant $I$ (either, a subring invariant or a permutation invariant), there is either a single-qubit gate, a generalized CNOT gate, or a $\tfour$ gate in $\class{G}$ that also violates the invariant.
\end{lemma}
\begin{proof}
Let $g \in G$ be the gate violating the invariant. There are two cases:
\begin{description}
\item[$I$ is a subring invariant:] Since $g$ violates the invariant, there is some entry $\Tstar(g)_{ij}$ of the tableau outside the subring. Apply Theorem~\ref{thm:gate_extraction} to this column of $g$ and it will extract either a generalized CNOT (if $\Tstar(g)_{ij}$ is non-invertible) or single-qubit gate (if $\Tstar(g)_{ij}$ is invertible) having the same entry, and thus violating the invariant. 
\item[$I$ is a permutation invariant:] The gate $g$ could violate the invariant in one of two ways. The may be more than one invertible entry in some column of $g$, in which case applying Theorem~\ref{thm:gate_extraction} extracts a $\tgate_{2k}$ for $k \geq 2$, which generates a $\tfour$, and $\tfour$ violates all permutations invariants. Otherwise, there is some entry of $\Tstar(g)_{ij}$ outside the allowed subring of elements, in which case we generate a single-qubit gate or generalized CNOT as in the other case.
\end{description}
\end{proof}

\begin{theorem}
\label{thm:final}
Let $G$ be any set of Clifford gates. If $\mathcal{S}$ is the smallest (with respect to lattice order) class of our classification containing $G$, then $\class{G} = \mathcal{S}$. 
\end{theorem}
\begin{proof}
There is a very general strategy for proving that the class $\class{G}$ is equal to the smallest $\mathcal{S}$ containing it.  Since $\class{G} \not \subseteq \mathcal{S}'$ for all $\mathcal{S'} \subsetneq \mathcal{S}$, there exists a gate $g \in G$ such that $g \notin \mathcal{S}'$.  That is, for each such class $\mathcal{S}'$, there is an invariant (described in Section~\ref{sec:invariants}) that $\class{G}$ fails to preserve.  Using the universal construction, we can generate a simple gate in $\class{G}$ which also fails to satisfy that invariant by Lemma~\ref{lem:invariant_to_gate}. For the purposes of this proof let us call this set of gates (i.e., the single-qubit gates, generalized CNOT gates, and $\tfour$ gates) the \emph{elementary gates}. Finally, we will show that these elementary gates generate $\mathcal{S}$ itself (sometimes requiring the identities from Section~\ref{sec:circuit_identities}).

We now give a complete sequence of tests to identify the class $\class{G}$.  It will be simpler to address the degenerate and non-degenerate gate classes separately, so let us assume for now that $\class{G}$ is non-degenerate; we will tackle the degenerate classes at the end. The rest of the case analysis depends on $\mathcal S$, the smallest class containing $G$.  At the highest level, we separate these classes based on which of the $X$-, $Y$-, $Z$-preserving invariants hold for $\mathcal{S}$. 

Suppose first that $\mathcal S$ is $X$-, $Y$-, and $Z$-preserving. There is only one non-degenerate class with these properties, so $\mathcal{S} = \class{\tfour,\P}$. Since $\class{G}$ is non-degenerate, there is some generalized CNOT gate or $\tfour$ gate in $\class{G}$ by Lemma~\ref{lem:invariant_to_gate}.  However, all generalized CNOT gates fail to preserve at least one of the bases and $\class{G} \subseteq \mathcal S$ is $X$-, $Y$-, and $Z$-preserving.  Therefore, the non-degenerate gate must be a $\tfour$. Since $\tfour$ generates the Pauli group (Lemma~\ref{lem:t4_generates_paulis}), we have that $\class{G} \supseteq \mathcal S$, and so $\class{G} = \mathcal S$.

Suppose now that $\mathcal S$ is $P$- and $Q$-preserving but not $R$-preserving for distinct Pauli operations $P$, $Q$, and $R$. Again, there is only one class with these properties, so $\mathcal{S} = \class{C(P,Q), \P}$. Therefore, there exists some gate $g \in G$ which is not in $\mathcal{S}' = \class{\tfour, \P}$.  In other words, $g$ fails to be $R$-preserving.  Note now there is no single-qubit gate which is both $P$- and $Q$-preserving but not $R$-preserving.  Furthermore, $C(P,Q)$ is the only generalized CNOT gate with this property. Therefore, by Lemma~\ref{lem:invariant_to_gate}, $g$ generates a $C(P,Q)$ gate. Clearly, $C(P,Q)$ generates both Pauli $P$ and $Q$, and therefore the whole Pauli group. Hence, $\class{G} = \mathcal{S}$.

Suppose now that $\mathcal S$ is $P$-preserving but not $Q$- and $R$-preserving. This is the most involved case and will require several more subdivisions. First suppose $\mathcal S = \class{C(P,Q), \R_P, \P}$.  We will use that $G$ generates an elementary gate which does not preserve the $Q$-basis ($\mathcal S' = \class{C(P,Q), \P}$), a gate which does not preserve the $R$-basis ($\mathcal S' = \class{C(P,R), \P}$), and a gate that is not $P$-orthogonal ($\mathcal S' = \class{\tfour, \R_P, \P}$).  Notice that all $P$-preserving single-qubit gates are also $P$-orthogonal, and recall that $C(P,P)$ is $P$-orthogonal.  Therefore, by Lemma~\ref{lem:invariant_to_gate}, $G$ generates $C(P,Q)$ or $C(P,R)$, both of which fail to be $P$-orthogonal.  Assume without loss generality it is $C(P,Q)$.  Let us now consider the gates that fail to be $Q$-preserving: if $G$ generates a $C(P,R)$ or $C(P,P)$ gate, then $G$ also generates a $\R_P$ gate by Lemma~\ref{lem:coalesce_cnots}; otherwise, $G$ generates a non-$Q$-preserving single-qubit gate, which must be $\R_P$ up to multiplication by Pauli elements.  Since $C(P,Q)$ generates the Pauli group, we get that $G$ generates an $\R_P$ gate, and so $\class{G} = \class{C(P,Q), \R_P, \P} = \mathcal S$. 

Suppose now that $\mathcal S$ is the $P$-orthogonal class $\class{\tfour, \R_P, \P}$.  First, $G$ generates a gate which is not $P$-degenerate ($\mathcal S' = \class{C(P,P), \P, \R_P}$). Since the only $P$-orthogonal single-qubit gates are also $P$-degenerate, we have that $G$ generates a $\tfour$ by Lemma~\ref{lem:invariant_to_gate} ($C(P,P)$ is $P$-degenerate and neither $C(P,Q)$ nor $C(P,R)$ are $P$-orthogonal).  Next, $G$ generates a gate which is not both $R$- and $Q$-preserving ($\mathcal S' = \class{\tfour, \P}$).   This is either a $C(P,P)$ gate or an $\R_P$ gate up to Pauli operations. Since $\tfour$ generates the Pauli group (Lemma~\ref{lem:t4_generates_paulis}) and $\tfour$ plus $C(P,P)$ generate an $\R_Z$ gate (Lemma~\ref{lem:tfour_cpp_to_rp}), then in either case $\R_P \in \class{G}$, implying that $\class{G} = \class{\tfour, \R_P, \P}$.

Suppose now that $\mathcal S$ is $P$-degenerate. Let us handle all five possible $P$-degenerate classes together.  First, since $\class{G}$ is non-degenerate, it must contain a $C(P,P)$ gate.  There are five $P$-degenerate classes correspond to the five $P$-preserving single-qubit classes containing $P$.  Unlike previous cases, such a diversity of classes exists because $C(P,P)$ does not suffice to generate the Pauli group on its own. We now use the normal form for $P$-degenerate gates given in Lemma~\ref{lem:semidegenerate}; that is, each gate can be expressed as a layer of single-qubit gates, a layer of $C(P,P)$ gates, and a permutation.  We can extract all the single-qubit gates (let's call them $G_1$) by canceling the $C(P,P)$ gates (with more $C(P,P)$ gates), canceling the SWAP gates (with more SWAP gates), and using eigenstate ancillas. We have that $\class{G} \subseteq \class{C(P,P), G_1}$ since every gate in $G$ is composed of such gates, and $\class{G} \supseteq \class{C(P,P), G_1}$ since we have shown how to generate all such gates.  Therefore, $\class{G} = \class{C(P,P), G_1}$.

Suppose now that $\mathcal S$ is the egalitarian class $\class{\tfour, \Gamma, \P}$, which is neither $P$-,$Q$-, nor $R$-preserving. Since $\tfour$ is the only egalitarian non-degenerate elementary gate, it can be generated in $\class{G}$. Furthermore, it must contain a single-qubit gate that is egalitarian, but not $P$-,$Q$-, or $R$-preserving ($\mathcal S' = \class{\tfour,\P}$).  The only such single-qubit gates are $\Gamma$ and $\Gamma^\dag$ up to Pauli operations. Therefore $\class{G} = \class{\tfour, \Gamma, \P}$ since $\tfour$ generates the Pauli group (Lemma~\ref{lem:t4_generates_paulis}).

Finally, suppose that $\mathcal{S}$ is the class of all Clifford gates, violating all invariants.  Let us organize the proof by which non-degenerate gates we generate from Lemma~\ref{lem:invariant_to_gate}. The first case is that we get no generalized CNOT gates, only a $\tfour$ gate. The immediate subclasses of $\mathcal{S}$ satisfy the $X$-, $Y$-, $Z$-preserving and egalitarian invariants, so $G$ generates a gate violating each of these invariants. Since this cannot be $\tfour$ (it satisfies all invariants), and we are assuming Lemma~\ref{lem:invariant_to_gate} gives us no generalized CNOT gates, these must be single-qubit gates. Clearly these single-qubit gates generate \emph{all} single-qubit gates, so by Lemma~\ref{lem:t4_plus_single_qubits_ftw}, $\class{G} = \ALL$.  

Another possibility is that Lemma~\ref{lem:invariant_to_gate} produces one or more generalized CNOT gates. Our main tool is Lemma~\ref{lem:cpq_plus_to_all}---if we have $C(P,Q)$, a non-$P$-preserving single-qubit gate, and a non-$Q$-preserving single-qubit gate (note: $P$ and $Q$ may be the same in this lemma!), then $\class{G} = \ALL$. In some cases, the generalized CNOT gates themselves can provide the single-qubit gates, so any collection of generalized CNOTs containing one of the following subsets generates $\ALL$.
\begin{itemize}
    \item $\{C(P,P),C(Q,Q)\}$: These gates generate $\theta_{P+Q}$ by Lemma~\ref{lem:coalesce_cnots}, which preserves neither the $P$ nor $Q$ basis. 
    \item $\{C(P,P),C(Q,R)\}$: These gates generate $\Gamma$ by Lemma~\ref{lem:coalesce_cnots}, which does not preserve any basis. 
    \item $\{C(P,Q),C(P,R),C(Q,R)\}$: These gates generate $\R_P$, $\R_Q$, and $\R_R$ by Lemma~\ref{lem:coalesce_cnots}, which together do not preserve any basis.
\end{itemize}
In fact, the only exceptions are collections of generalized CNOT gates such that all gates are $P$-preserving (for at least one Pauli basis $P$), so let us suppose Lemma~\ref{lem:invariant_to_gate} gives us only $P$-preserving generalized CNOT gates and a non-$P$-preserving single-qubit gate. We are done if one of the generalized CNOT gates is $C(P,P)$ by Lemma~\ref{lem:cpq_plus_to_all}---otherwise, we have one of the following sets:
\begin{itemize}
    \item $\{C(P,Q)\}$:  $G$ generates a single-qubit gate which is not $P$-preserving and a single-qubit gate which is not $Q$-preserving by Lemma~\ref{lem:invariant_to_gate}. 
    \item $\{C(P,Q),C(P,R)\}$: $G$ generates a single-qubit gate that is not $P$-preserving by Lemma~\ref{lem:invariant_to_gate}. We note that any single-qubit gate which fails to be $P$-preserving must also fail to be either $Q$- or $R$-preserving. 
\end{itemize}
This finishes the $\mathcal{S} = \ALL$ case analysis. 

Let us now return to the degenerate classes where the argument will be similar to the $P$-degenerate gate classes.  We will associate every degenerate class with a subgroup of the single-qubit Clifford gates.  Explicitly, we decompose every degenerate gate into a circuit of single-qubit gates and SWAP gates by Lemma~\ref{lem:degenerate}, and then extract all the single-qubit gates by canceling the SWAP gates (with more SWAP gates), and using eigenstate ancillas.  Therefore, $\class{G}$ is associated with a class of single-qubit gates, for which the classification is well known (see Figure~\ref{fig:degeneratelattice}).
\end{proof}

\begin{corollary}
Given any set of gates $G$, there is a subset $S \subseteq G$ of at most three gates such that $\class{S} = \class{G}$.
\end{corollary}
\begin{proof}
Suppose to the contrary that there exists a set of gates $G$ such that $\class{G} \neq \class{S}$ for any subset $S \subseteq G$ of size $\leq 3$. Without loss of generality, let $G$ be a minimal set of gates with this property. If there are subsets $A, B \subseteq G$ such that $\class{A} \subseteq \class{B}$ but $A \not \subseteq B$ then we could delete $A \backslash B$ from $G$ and it will generate the same class, but have strictly fewer subsets of size $3$. It follows that in a minimal $G$, all subsets of $G$ generate distinct classes. But then we can further shrink $G$ to any $4$-element subset, say $G = \{ g_1, g_2, g_3, g_4 \}$, since that is distinct from all size $\leq 3$ subsets. 

Now we think of $\class{\cdot}$ as a map embedding the power set lattice for $4$ elements into our class lattice. We rule this out by case analysis. The first case is when $\class{G}$ is degenerate. Notice that an ascending chain like 
$$
\varnothing \subseteq \{ g_1 \} \subseteq \{ g_1, g_2 \} \subseteq \{ g_1, g_2, g_3 \} \subseteq \{ g_1, g_2, g_3, g_4 \}
$$
must map to a similar ascending chain in the degenerate class lattice (Figure~\ref{fig:degeneratelattice}). All such ascending chains go through $\P + \Gamma$, $\P + \R_X$, $\P + \R_Y$, or $\P + \R_Z$, and end at $\top$, so the 4 distinct classes $\class{g_1, g_2, g_3}, \ldots, \class{g_2, g_3, g_4}$ must correspond to those 4 cases. Let's say $\class{g_1, g_2, g_3} = \class{\P + \Gamma}$. However, there are $6$ ascending chains through $\{ g_1, g_2, g_3 \}$, and only $3$ chains going through $\P + \Gamma$, so two subsets of $G$ map to identical classes, a contradiction. 

The more involved case is when $\class{G}$ is non-degenerate. One of the gates must generate a non-degenerate class (otherwise $\class{G}$ would be degenerate). Again, we consider ascending chains in the lattice. For instance, $\class{g_1}, \ldots, \class{g_4}$ cannot contain $C(P,Q)$ (for Paulis $P \neq Q$) because the chains above those classes are insufficiently long (i.e., at best two hops up to $\textsf{ALL}$---we need three). With the exception of $\textsf{ALL}$, no class containing $C(P,Q)$ is the join of two classes \emph{not} containing $C(P,Q)$. Therefore no proper subset of $G$ generates any $C(P,Q)$ gate.

Now imagine removing classes (except $\textsf{ALL}$) which contain $C(P,Q)$. With those classes gone, the ascending chains above the classes containing $\tfour$ are too short for any of them to be a generator, just like we argued for $C(P,Q)$. Similarly, these $\tfour$-containing classes (excluding $\textsf{ALL}$) are not the join of any classes \emph{not} containing $\tfour$. Once we remove \emph{those}, the only remaining non-degenerate classes contain $C(P,P)$ for some $P$, so suppose $C(P,P) \in \class{g_1}$. If $\class{g_2}$ is not also $P$-preserving, then we see from the lattice that $\class{g_1, g_2} = \textsf{ALL} = \class{g_1, g_2, g_3}$, a contradiction. Hence, all of the classes are $P$-preserving. But there are only $5$ $P$-preserving classes left in the lattice, and $7$ distinct classes generated by $g_1$ and proper subsets of $\{ g_2, g_3, g_4 \}$. Hence, $\class{G}$ cannot be non-degenerate either. 

We conclude that for any gate set $G$, there is a subset of at most $3$ gates that generates it. This is tight, since $C(X,Y)$, $C(Y,Z)$, and $C(Z,X)$ generate all Clifford gates, but any pair of them cannot generate the third. 
\end{proof}

%!TEX root = ../full_paper.tex

\section{Open Problems}
\label{sec:open_problems}

Our classification of Clifford gates resolves an open problem of Aaronson et al.\ \cite{ags:2015}, but leaves their central question, the classification of arbitrary quantum gates, completely open. It is unclear whether there is another piece of the full quantum gate classification that can be peeled off. Other discrete quantum gate sets are known, but none are known to have the rich structure and entanglement of Clifford gates (aside from conjugated Clifford gates). So we ask: are there other interesting discrete gate sets, and can they be classified like Clifford gates?

Another source of open problems is the choice of ancilla rule. As discussed, we permit ancillas initialized to arbitrary quantum states. We have determined that the classification continues to hold under a stabilizer state ancilla model if the following conjecture holds:
\begin{conjecture}
\label{conj:single_qubit_stabilizers}
For any single-qubit Clifford gate $g$, there exists a stabilizer state $\ket{\psi}$ and circuit of $\operatorname{SWAP}$ gates $\pi$ such that $g \circ \pi \ket{\psi} = \ket{\psi}$. 
\end{conjecture}
This is sufficient to remove single-qubit gates in situations where we would otherwise use an eigenstate.

For many single-qubit gates, there is a trivial stabilizer eigenstate. For instance, $X$ stabilizes $\ket{+}$, $\R_Z$ stabilizes $\ket{0}$, and many other single-qubit Clifford gates are conjugate to one of these cases. Now consider the gate $\theta_{X+Z}$, whose eigenstates (unnormalized) $(1 \pm \sqrt{2}) \ket{0} + \ket{1}$ are \emph{not} stabilizer states. How then, given the gate $\theta_{X+Z} \otimes \theta_{X+Z}$, does one generate the gate $\theta_{X+Z}$ which acts only on one qubit?  Han-Hsuan Lin discovered the first explicit nine qubit stabilizer state for this task.

Let $\pi$ be a circuit that cyclically permutes qubits 2 through 9, and suppose $\theta_{X+Z}$ is applied to qubit 1. Let $\ket{\psi}$ be the state stabilized by the following commuting Pauli strings,
\begin{center}
\begin{tabular}{l l l l}
$X XZXZIIII$,  & $Z IXZXZIII$, &  $X IIXZXZII$, & $Z IIIXZXZI$, \\
$X IIIIXZXZ$, & $Z ZIIIIXZX$, & $X XZIIIIXZ$, & $Z ZXZIIIIX$, \\
$Y YIIIYIII$,  & $-Y IYIIIYII$, & $Y IIYIIIYI$, & $-Y IIIYIIIY$,
\end{tabular}
\end{center}
9 of which are independent.  One can check that conjugating each generator by $\theta_{X+Z} \circ \pi$ yields another element of the stabilizer group, so $(\theta_{X+Z} \circ \pi )\ket{\psi} = \ket{\psi}$. In other words, Conjecture~\ref{conj:single_qubit_stabilizers} holds for $\theta_{X+Z}$, and for all conjugates $\theta_{P+Q}$ by symmetry.

% All that remains to verify the conjecture is to find a similar state fixed by the eight remaining gates---the $\Gamma$ gate and its conjugates. 
To verify the conjecture, it suffices to find a stabilizer state $\ket{\psi}$ and circuit $C$, constructed of SWAP gates and a single $\Gamma$ gate, such that $C \ket{\psi} = \ket{\psi}$.

\section{Acknowledgments}
Much on this work was completed at MIT, where both authors were graduate students. DG additionally acknowledges the support an NSF Graduate Research Fellowship under Grant No.~1122374. We would like to thank Scott Aaronson for his guidance throughout this project.  Also, thanks to Han-Hsuan Lin and Adam Bouland for useful discussions. Finally, we thank the many anonymous reviewers that have helped to improve the manuscript. In particular, we thank the reviewer that pointed out an error in an attempted proof of Conjecture~\ref{conj:single_qubit_stabilizers}.

\bibliographystyle{quantum}
\bibliography{bibliography}

\appendix
%!TEX root = ../full_paper.tex

\section{Enumeration}
\label{app:enumeration}

\begin{theorem}
\label{thm:enumeration}
Let $\# \class{\cdot}_n$ denote the number of $n$-qubit gates in a class. Then 
\begin{align*}
\# \class{G}_n &= \abs{G}^n n! && \text{for $G$ a group of single-qubit gates}, \\
\# \class{C(Z,Z), G}_n &= \abs{G}^n 2^{n(n-1)/2} n! && \text{for $\class{Z}_1 \subseteq G \subseteq \class{\P, \R_Z}_1$ a group}, \\
\# \class{C(Z,X), \P}_n &= 4^n 2^{n(n-1)/2} \prod_{i=1}^{n} (2^i - 1), \\
\# \class{C(Z,X), \P, \R_Z}_n &= 8^n 2^{n(n-1)} \prod_{i=1}^{n} (2^i - 1), \\
\# \class{\tfour, \P}_n &= 4^{n} a(n), \\
\# \class{\tfour, \P, \R_Z}_n &= 8^{n} 2^{n(n-1)/2} a(n), \\
\# \class{\tfour, \P, \Gamma}_n &= 4^n 2^{n(n-1)/2} \prod_{i=1}^{n} (2^{i} - (-1)^{i}), \\
\# \class{\ALL}_n &= 4^n 2^{n^2} \prod_{i=1}^{n} (4^i - 1),
\end{align*}
where
$$
a(n) = \begin{cases}
2^{m^2} \prod_{i=1}^{m-1} (2^{2i} - 1), & \text{if $n = 2m$,} \\
2^{m^2} \prod_{i=1}^{m} (2^{2i} - 1), & \text{if $n = 2m+1$.} \\
\end{cases}
$$
\end{theorem}
\begin{proof}
Most of these numbers follow from the lemmas above. For example, consider the class $\class{C(Z,X), \P, \R_Z}$. It follows from Lemma~\ref{lem:eliminate} that any gate in this class has a circuit consisting of a layer of $C(Z,X)$ gates, then a layer of $C(Z,Z)$ gates, then a layer of single-qubit gates in $G$. 

We would like to count the number of possible gates by multiplying the number of possibilities for each layer, but we must be careful that there is no gate with two circuit representations. Suppose for a contradiction that $g_1$ and $g_2$ generate the same gate, but some layer of $g_1$ differs from $g_2$. Then $g_1^{-1} g_2$ is the identity, since $g_1$ and $g_2$ generate the same transformation.

On the other hand, the $C(Z,X)$ layers of $g_1$ and $g_2$ meet in the middle of the circuit for $g_1^{-1} g_2$. If those layers do not generate the same linear transformation, then the combination is some non-trivial linear transformation which is, in particular, not $Z$-degenerate. The other layers of $g_1$ and $g_2$ \emph{are} $Z$-degenerate, so we conclude that $g_1^{-1} g_2$ is not $Z$-degenerate (if it were, we could invert the outer layers to show that the two middle layers are $Z$-degenerate). But $g_1^{-1} g_2 = I$ is clearly $Z$-degenerate, therefore the $C(Z,X)$ layers of $g_1$ and $g_2$ must generate the same linear transformation.

The $C(Z,X)$ layers of $g_1$ and $g_2$ cancel (since we have shown they are equivalent), so they effectively disappear, and we make a similar argument about the $C(Z,Z)$ layers, and then the single-qubit layers. That is, if the $C(Z,Z)$ layers do not contain the same set of $C(Z,Z)$ gates, then we obtain a contradiction because they produce a non-degenerate layer in the middle, implying that $g_1^{-1} g_2 = I$ is non-degenerate. Once we remove the $C(Z,Z)$ layers, the single-qubit layers must be the same or they would leave behind a non-trivial single-qubit gate. We conclude that all layers of $g_1$ and $g_2$ are actually the same, so the number of gates is the product of the number of choices for each layer. 

Now the problem is to count the number of choices for each layer. For the single-qubit layer, this is clearly just $n$ independent choices of single-qubit gate from $\class{\P, \R_Z}$, or $8^n$. For the $C(Z,Z)$ layer, there is a choice whether or not to place a $C(Z,Z)$ gate in each of the $\binom{n}{2}$ possible positions, so $2^{n(n-1)/2}$ choices for the layer. For the $C(Z,X)$ layer, observe that $C(Z,X)$ generate precisely the set of invertible linear transformations, of which there are 
$$
2^{n(n-1)/2} \prod_{i=1}^{n} (2^i - 1)
$$
by a classical argument. Multiplying the three layers, we have a total of
$$
\# \class{C(Z,X),\P,\R_Z}_n = 8^n 2^{n(n-1)} \prod_{i=1}^{n} (2^i - 1)
$$
$n$-qubit transformations generated by $C(Z,X)$, $\P$, and $\R_Z$. 

The numbers for $\class{G}$, $\class{C(Z,Z), G}$, $\class{C(Z,X),\P}$, $\class{\tfour, \P}$, and $\class{\tfour, \P, \R_Z}$ follow by a similar argument, although for the last two classes we need the fact that $\tfour$ generates 
$$
a(n) = \begin{cases}
2^{m^2} \prod_{i=1}^{m-1} (2^{2i} - 1), & \text{if $n = 2m$,} \\
2^{m^2} \prod_{i=1}^{m} (2^{2i} - 1), & \text{if $n = 2m+1$.}
\end{cases}
$$ 
orthogonal transformations on $n$ qubits. 

For the final two classes, we use known expressions (see, e.g., equations (2.6.1) and (2.6.2) from \cite{wall:1963}) for the number of $n \times n$ unitary matrices over $\mathbb{F}_4$ (in the case of $\class{\tfour, \P, \Gamma}$) and for the number of $2n \times 2n$ symplectic matrices over $\ftwo$ (in the case of $\ALL$). We multiply by $4^{n}$ in both cases to account for the phase bits, which are completely independent of the matrix part. 
\end{proof}

\begin{theorem}
The asymptotic size of each class is as follows.
\begin{align*}
\log_2 \# \class{G}_n &= n \log_2(\abs{G}) + n\log_{2}n- {n}{\log_2 e}+\frac{1}{2}\log_{2}2\pi+O\left(  \frac{1}{n}\right), \\
\log_2 \# \class{C(Z,Z), G}_n &= n \log_2(\abs{G}) + \frac{n(n-1)}{2} + n\log_{2}n- {n}{\log_2 e}+\frac{1}{2}\log_{2}2\pi+O\left(  \frac{1}{n}\right), \\
\log_2 \# \class{C(Z,X), \P}_n &= n^2 + 2n -\alpha + O(2^{-n}), \\
\log_2 \# \class{C(Z,X), \P, \R_Z}_n &= \frac{3}{2}n^2 + \frac{5}{2}n -\alpha + O(2^{-n}), \\
\log_2 \# \class{\tfour, \P}_n &= \frac{1}{2}n^2 +\frac{3}{2}n -\beta + O(2^{-n}), \\
\log_2 \# \class{\tfour, \P, \R_Z}_n &=  n^2 + 3n -\beta + O(2^{-n}), \\
\log_2 \# \class{\tfour, \P, \Gamma}_n &= n^2 + 2n + \gamma  + O(2^{-n}),\\
\log_2 \# \class{\ALL}_n &= 2n^2 + 3n - \beta + O(4^{-n}).
\end{align*}
where $G$ is the same as in Theorem~\ref{thm:enumeration}, and 
\begin{align*}
\alpha &  =-\sum_{i=1}^{\infty}\log_{2}(1-2^{-i})\approx1.7919,\\
\beta &  =-\sum_{i=1}^{\infty}\log_{2}(1-4^{-i})\approx0.53839, \\
\gamma & = \sum_{i=1}^{\infty}\log_{2}(1-(-2)^{-i}) \approx 0.27587.
\end{align*}
\end{theorem}
\begin{proof}
We take the logarithm of each class size, which we can separate into the logarithm of each layer comprising that class, as in Theorem~\ref{thm:enumeration}.  For most layers this is straightforward, except for the layer of permutations, orthogonal transformations, or general linear transformations.  The first we handle with Stirling's approximation.  For the other two, we factor out powers of two leaving a partial sum of a convergent series, which we analyze with a Taylor expansion.  The classes $\class{\tfour, \P, \Gamma}_n$ and $\class{\ALL}_n$ follow by similar techniques.
\end{proof}

\begin{corollary}
Let $\mathcal{C}$ be any class, and let $G$\ be an $n$-qubit
gate chosen uniformly at random from $\mathcal{C}$.  Then%
\[
\Pr\left[G \text{ generates }\mathcal{C}\right]  =1-O(2^{-n}).
\]
\end{corollary}

%!TEX root = ../full_paper.tex

\section{Classical reversible gates with quantum ancillas}
\label{app:classical_quantum}
In this section we describe what the classical reversible gate lattice of Aaronson et al.\ \cite{ags:2015} would look like under quantum ancillas.  We extend all classical gates discussed in that paper to the quantum setting in the most natural way.  Figure~\ref{fig:classical_quantum_lattice} shows the new (dramatically simpler) lattice.

Some of the collapses in the lattice are immediate.  For instance, the class $\class{\NOT \otimes \NOT}$ collapses with the class $\class{\NOT}$ because $\NOT \ket{+} = \ket{+}$.  A similar collapse occurs been all classes where the parity issue arises, such as between the classes $\class{\CNOTNOT}$ and $\class{\CNOT}$.

A more interesting collapse occurs between all mod-$k$-preserving classes for $k \ge 2$.  Consider the following gate $G : \{0,1\}^k \rightarrow \{0,1\}^k$ of order 2 which preserves Hamming weight mod $k$:
\begin{align*}
G(0^k) &= 1^k \\
G(1^k) &= 0^k \\
G(1^a 0^b 0) &= 1^{a-1} 0^{b+1} 1 \\
G(1^a 0^b 1) &= 1^{a+1} 0^{b-1} 0
\end{align*}
where $G$ acts as the identity on all other inputs.  Since $G$ preserves the Hamming weight mod $k$, it must appear in the class.  We will show that $G$ can generate a $\NOT$ gate.  Let 
$$\ket{\psi_k} := \frac{1}{\sqrt{k}} \sum_{i=0}^{k-1} \ket{1^i 0^{k-i-1}}$$
so, for example
$$\ket{\psi_4} = \frac{\ket{000} + \ket{100} + \ket{110} + \ket{111}}{2}.$$
Now, for $b \in \{0,1\}$, $G(\ket{\psi_k}\ket{b}) = \ket{\psi_k}\ket{b\oplus 1}$.  Therefore, each mod-$k$-preserving class for $k \ge 2$ collapses to the $\class{\Fredkin, \NOT}$ class.  Furthermore,  Figure~\ref{fig:CNOT_from_Fredkin_NOT} shows that the Fredkin and NOT gates are sufficient to generate a $\CNOT$. Therefore, every non-conservative non-affine class generates all classical reversible transformations.  

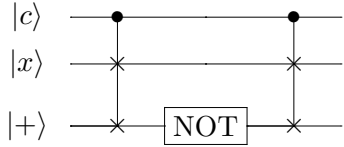
\begin{figure}
\centering
\begin{minipage}[c]{.3\textwidth}
\Qcircuit @C=1.5em @R=1.5em {
\ket{c} & & \ctrl{2} & \qw & \ctrl{2} & \qw \\
\ket{x} & & \qswap &  \qw & \qswap & \qw \\
\ket{+} & & \qswap \qw & \gate{\NOT} & \qswap \qw & \qw
}
\end{minipage}
\caption{Generating CNOT from Fredkin and NOT gates}
\label{fig:CNOT_from_Fredkin_NOT}
\end{figure}

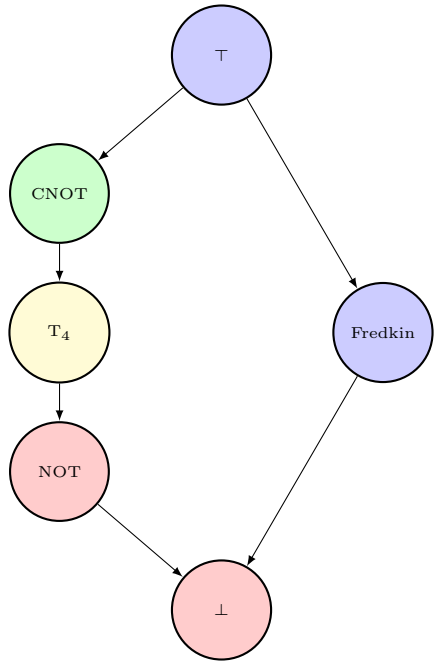
\begin{figure}
\begin{center}
\begin{tikzpicture}[>=latex]
\tikzstyle{class}=[circle, thick, minimum size=1.2cm, text width=1.0cm, align=center, draw, font=\tiny]
\tikzstyle{nonaffine}=[class, fill=blue!20]
\tikzstyle{affine}=[class, fill=green!20]
\tikzstyle{orthogonal}=[class, fill=yellow!20]
\tikzstyle{none}=[class,fill=red!20]
\matrix[row sep=0.5cm,column sep=0.8cm] {
& \node (ALL) [nonaffine]{$\top$}; &\\
\node (CNOT) [affine]{$\CNOT$}; & & \\
\node (T4) [orthogonal]{$\tfour$}; & & \node (FREDKIN) [nonaffine]{$\Fredkin$}; \\
\node (NOT) [none]{$\operatorname{NOT}$}; & & \\
& \node (NONE) [none]{$\bot$}; & \\
};
% EDGES
\path[draw,->] 
(CNOT) edge (T4)
(T4) edge (NOT)
(NOT) edge (NONE)
(ALL) edge (CNOT)
(ALL) edge (FREDKIN)
(FREDKIN) edge (NONE);
\end{tikzpicture}
\end{center}
\caption{The inclusion lattice of classical gates using quantum ancillas.}
\label{fig:classical_quantum_lattice} 
\end{figure}

We now only need to prove that the classes appearing in Figure~\ref{fig:classical_quantum_lattice} are distinct.  Notice, however, that the classes $\class{\CNOT}$, $\class{\tfour}$, and $\class{\NOT}$ all have Clifford gate generators, which by the results of this paper, generate distinct classes.  We only need to show that the $\class{\Fredkin}$ class is distinct from the remaining classes.  However, the invariant in \cite{ags:2015} more or less functions to prove this separation.  Namely, Fredkin conserves the Hamming weight of its input.  Therefore the sum of the Hamming weights of the computational basis states of the input state is conserved.  However, the NOT gate necessarily changes this sum, witnessing that $\NOT \notin \class{\Fredkin}$, and therefore that the lattice is complete.

%!TEX root = ../full_paper.tex

\section{Three-qubit generator for \texorpdfstring{$\class{\tfour, \Gamma, \P}$}{<T4,Gamma,P>}}
\label{app:gta}

The $\tfour$ gate is a minimal generator for a class of orthogonal gates, both in our classification (i.e., $\class{\tfour, \P}$), and in the Aaronson et al.\ \cite{ags:2015} classification. Surprisingly, when we add $\Gamma$ gates to this class, it has a \emph{three}-qubit generator. This is most easily seen by counting (using the enumeration results from Appendix~\ref{app:enumeration}):
\begin{align*}
&\# \class{\tfour, \P, \Gamma}_3 = 2^{3(3-1)/2+2(3)} \prod_{i=1}^{3} (2^{i} - (-1)^{i}) = 41472, \\
&\# \class{\Gamma, \P}_3 = 12^3 3! = 10368, \\
&\# \class{\tfour, \P}_3 = 4^{3} 2^{1^2} \prod_{i=1}^{1} (2^{2i} - 1) = 384, \\
&\# \class{\P}_3 = 4^3 3! = 384.
\end{align*}
We see that $\class{\tfour, \P}$ and $\class{\P}$ have the same number of gates on three qubits, but $\class{\tfour, \P, \Gamma}$ has substantially more gates than $\class{\Gamma,\P}$. Notice that there are 4 cosets of $\class{\Gamma, \P}_3$ in $\class{\tfour, \P, \Gamma}_3$ by Lagrange's Theorem, corresponding to 4 gates that are nonequivalent up to applications of elements in $\class{\Gamma, \P}_3$.  If we let $\alpha = \left( \begin{smallmatrix}0 & 1 \\ 1 & 1 \end{smallmatrix} \right) \in \ring$, then one such gate is described by the following tableau
$$\left(\begin{array}{ccc|c} \alpha & I  & I  & 0\\ I & \alpha & I & 0\\ I & I & \alpha & 0 \end{array}\right) .$$
This is equal to the gate $O_3$ (up to single-qubit gates) that appears in Section~\ref{sec:universal_construction}. By Theorem~\ref{thm:final}, this gate generates all of $\class{\tfour, \P, \Gamma}$.

%!TEX root = ../full_paper.tex

\section{Canonical form from Section~\ref{sec:universal_construction}}
\label{app:canonical_formTODO}

Recall the decomposition for an $n$-qubit Clifford gate $G$ in Lemma~\ref{lem:decomp}: 
\begin{equation}
    G = (P \otimes D) \circ \mathrm{SWAP}(1,i) \circ \left( \prod \CNOT(r_j) \right) \circ O_{2k-1} \circ \left( \bigotimes \G(r_j) \right).
    \label{eqn:decomp}
\end{equation}
It decomposes as single-qubit gates, an $O_{2k-1}$ gate, generalized CNOT gates, an optional SWAP, a Pauli on the first qubit, and an arbitrary Clifford gate $D$ on the remaining $(n-1)$ qubits. Conceptually, the decomposition is a mapping $\mu \colon \CC_n \to \CC_{n-1} \times Q_n$, taking an $n$-qubit Clifford to an $(n-1)$-qubit Clifford (i.e., $D$) and some subset of Clifford gates $Q_n$ (i.e., all gates except $D$ in the decomposition). In this appendix we will argue that $\mu$ is actually a bijection, leading to a nice canonical form for Clifford gates. 

First, $\mu$ is clearly one-to-one since we can multiply the decomposition out to recover the original Clifford $G$. To show that $\mu$ is bijective, we simply need to argue that the domain and codomain have the same size. We conveniently have an expression for the number of Clifford operations from Appendix~\ref{app:enumeration}, and can compute the ratio 
$$
\frac{|\CC_{n}|}{|\CC_{n-1}|} = \frac{4^n 2^{n^2} \prod_{i=1}^{n} (4^i - 1)}{4^{n-1} 2^{(n-1)^2} \prod_{i=1}^{n-1} (4^i - 1)} = 2^{2n+1} (4^{n} - 1).
$$
\begin{lemma}
The decomposition produces $|Q_n|$ possible $n$-qubit Clifford circuits (excluding $D$) where 
$$
|Q_n| = 2^{2n+1} (4^{n} - 1).
$$
\end{lemma}
\begin{proof} We will count the number of configurations for the gates in each layer of the decomposition.  A key parameter is the number of invertible elements in the first column of $\Tstar(G)$, which determines the number of single-qubit gates we apply.  Let $S$ be the set of indices for these invertible elements with $s := |S|$.

There are $6$ single-qubit gates (corresponding to the $6$ invertible elements of $\ring$), and we apply one such gate to each qubit in $S$ for a total of $6^s$ possible single-qubit layers given $S$. Next, we consider the non-invertible elements in the column, corresponding to the $10$ choices of generalized CNOT gate (including the identity gate).  Since we apply a CNOT targeting each qubit \emph{not} in $S$, this gives $10^{n-s}$ choices for this layer. Finally, we have four choices for a Pauli $P$. The remaining degrees of freedom in the circuit all have a canonical choice: take $i$ as the first invertible element in $S$, perform the generalized CNOT gates (which may not commute) in the order of their target (not $i$), and likewise for the single-qubit gates (except they \emph{do} commute). The targets of the SWAP are determined by $i$, and we omit it if and only if $i=1$. The targets of the $O_{s}$ gate are $S$, and the orientation is fixed by $i$.

Now we sum over $s$ and count the number of choices for $S$ (i.e., $\binom{n}{s}$), the single-qubit gates ($6^{s}$), and generalized CNOT gates ($10^{n-s})$).
\begin{align*}
    |Q_n| &= 4 \cdot \sum_{\text{$s$ odd}} \binom{n}{s} 6^{s} 10^{n-s} \\
    &= 2 \cdot \left( \sum_{s=0}^{n} \binom{n}{s} 6^{s} 10^{n-s} - \sum_{s=0}^{n} \binom{n}{s} (-6)^{s} 10^{n-s} \right) \\
    &= 2 \left( (6 + 10)^{n} - (10 - 6)^{n} \right) \\
    &= 2^{2n+1} (4^{n} - 1).
\end{align*}
\end{proof}
It follows that $|\CC_{n}|/|\CC_{n-1}| = 2^{2n+1} (4^n - 1) = |Q_n|$, so the decomposition is a bijection.

%!TEX root = ../full_paper.tex

\section{Canonical form for 2-qubit circuits}
\label{app:canonical_form}

In this section, we describe a very clean canonical form for 2-qubit Clifford circuits.  

\begin{theorem}
Let $G$ be any Clifford circuit on two qubits.  Then, $G$ is equivalent to a circuit of at most depth 3 composed of the following sequence of gates
\begin{enumerate} \itemsep 0px \parskip 0px
\item a $\SWAP$, and
\item a tensor product of single-qubit gates, and 
\item a generalized CNOT gate,
\end{enumerate}
where we can choose at each step whether or not to include the gate.  That is, $G$ is of the form of the circuit depicted in Figure~\ref{fig:2_qubit_canonical}.
\end{theorem}
\begin{proof}
Since $G$ is a Clifford circuit, it can be written as a product of CNOT, $\theta_{X+Z}$ (Hadamard), and $\R_Z$ gates.  Recall that conjugating a generalized CNOT gate by a single-qubit gate is simply another generalized CNOT gate.  Therefore, we can push all the single-qubit gates left and all the generalized CNOT gates right.  All that remains to show is that we can coalesce the generalized CNOT gates into a single CNOT gate.  We refer to Table~\ref{table:coalescing_rules} for those equivalences, and note that identical generalized CNOT gates cancel.  Eventually, what remains is a circuit composed of single-qubit gates, SWAP gates, and at most one generalized CNOT gate.  We can push the SWAP gates to the left (they collapse to either a single SWAP gate or the identity) and combine the single-qubit gates, which completes the proof.   
\end{proof}
% \footnote{Notice that we don't need to include SWAP gates because of the equivalence $\SWAP(1,2) = \CNOT(1,2) \circ \CNOT(2,1) \circ \CNOT(1,2)$.}

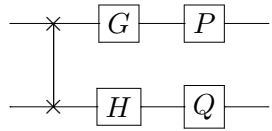
\begin{figure}
$$\Qcircuit @C=1.5em @R=1.5em {
 & \qswap  & \gate{G} & \gate{P}& \qw \\
 & \qswap \qwx  \qwx & \gate{H} & \gate{Q} \qwx & \qw}$$
\caption{Canonical form of a 2-qubit circuit: optional SWAP gate, single-qubit gates $G$ and $H$, and $C(P,Q)$ gate.}
\label{fig:2_qubit_canonical}
\end{figure}

\begin{table}
\centering
\newcommand{\raiseq}{-15px}
\newcommand{\cnotcolumn}{1.5em}
\newcommand{\cnotrow}{1.5em}
\begin{align*}
\Qcircuit @C=\cnotcolumn @R=\cnotrow {
& \gate{P} & \gate{P} & \qw  \\
& \gate{R} \qwx & \gate{Q} \qwx & \qw}
   &\raisebox{\raiseq}{=}
\Qcircuit @C=\cnotcolumn @R=\cnotrow {
& \gate{\R_{P}} & \gate{P} & \qw  \\
& \qw & \gate{P} \qwx & \qw } \\ %% NEXT
\Qcircuit @C=\cnotcolumn @R=\cnotrow {
& \gate{P} & \gate{Q}& \qw  \\
& \gate{P} \qwx & \gate{R} \qwx &  \qw}
   &\raisebox{\raiseq}{=}
\Qcircuit @C=\cnotcolumn @R=\cnotrow {
& \qswap  & \gate{\Gamma}  & \gate{P} & \qw \\
& \qswap \qwx  & \gate{\Gamma^\dag} & \gate{P} \qwx & \qw} \\ %% NEXT
\Qcircuit @C=\cnotcolumn @R=\cnotrow {
& \gate{P}  & \gate{Q}   & \qw  \\
& \gate{P} \qwx & \gate{Q} \qwx & \qw}
   &\raisebox{\raiseq}{=}
\Qcircuit @C=\cnotcolumn @R=\cnotrow {
& \qswap  & \gate{\theta_{P+Q}}  & \gate{P}& \qw  \\
& \qswap \qwx & \gate{\theta_{P+Q}} & \gate{P} \qwx & \qw} \\ %% NEXT
\Qcircuit @C=\cnotcolumn @R=\cnotrow {
& \gate{P}  & \gate{Q}   & \qw  \\
& \gate{Q} \qwx & \gate{R} \qwx & \qw}
   &\raisebox{\raiseq}{=}
\Qcircuit @C=\cnotcolumn @R=\cnotrow {
& \qswap  & \gate{\R_Q}  & \gate{P}& \qw  \\
& \qswap \qwx & \gate{\R^\dag_Q} & \gate{Q} \qwx & \qw} \\ %% NEXT
\Qcircuit @C=\cnotcolumn @R=\cnotrow {
& \gate{P}  & \gate{Q}   & \qw  \\
& \gate{Q} \qwx & \gate{P} \qwx & \qw}
   &\raisebox{\raiseq}{=}
\Qcircuit @C=\cnotcolumn @R=\cnotrow {
& \qswap  & \gate{P}& \qw  \\
& \qswap \qwx & \gate{Q} \qwx & \qw} %% NEXT
%\Qcircuit @C=\cnotcolumn @R=\cnotrow {
%& \gate{P}  & \gate{P}   & \qw  \\
%& \gate{P} \qwx & \gate{P} \qwx & \qw}
%   &\raisebox{\raiseq}{=}
%\Qcircuit @C=\cnotcolumn @R=\cnotrow {
%& \qw & \qw  \\
%& \qw & \qw} \\ %% NEXT
%\Qcircuit @C=\cnotcolumn @R=\cnotrow {
%& \gate{P}  & \gate{P}   & \qw  \\
%& \gate{R} \qwx & \gate{R} \qwx & \qw}
%   &\raisebox{\raiseq}{=}
%\Qcircuit @C=\cnotcolumn @R=\cnotrow {
%& \qw &\ghost{\gate{P}} \qw & \qw  \\
%& \qw  & \ghost{\gate{R}} \qw & \qw} 
\end{align*}
\caption{Rules for coalescing generalized CNOT gates, assuming $\Gamma P \Gamma^\dag =  Q$ and $\Gamma Q \Gamma^\dag =  R$.}
\label{table:coalescing_rules}
\end{table}

\end{document}